\documentclass[letterpaper, 10 pt, conference]{ieeeconf}  

\usepackage{color,array}

\usepackage{amssymb,amsfonts}
\usepackage{algorithm}
\usepackage{algpseudocode}
\usepackage{float}
\usepackage{subcaption}

\usepackage{graphicx}
\usepackage{verbatim}   
\usepackage{booktabs}
\usepackage{algorithm}
\usepackage{algpseudocode}
\usepackage{subcaption}
\usepackage{graphicx}

\usepackage{mathtools}
\usepackage{siunitx}
\usepackage{tabularx}
\usepackage{pgfplots}
\usepackage{pgfplotstable}

\usepackage{multirow}
\usepgfplotslibrary{fillbetween}
\usepackage{booktabs}
\pgfplotsset{compat = newest}
\usetikzlibrary{arrows.meta,positioning,shapes,intersections,patterns,calc,fit,decorations,decorations.markings}

\usepackage{cite}

\newtheorem{rem}{Remark}
\newtheorem{thm}{Theorem}

\newtheorem{cor}{Corollary}
\newtheorem{prop}{Proposition}

\newtheorem{assum}{Assumption}

\newcommand\tran{\mkern-2mu\raise1.25ex\hbox{$\scriptscriptstyle\top\hspace{0.5mm}$}\mkern-3.5mu}
\newcommand{\R}{\mathbb{R}}

\newcommand{\N}{\mathbb{N}}

\newcommand{\bm}[1]{{\boldsymbol{#1}}}

\DeclareMathOperator{\Prob}{P}

\newcommand{\x}{\bm x}

\newcommand{\f}{\bm{f}}

\newcommand{\hff}{\bm{h}_{\textnormal{ff}}}
\newcommand{\hfb}{\bm{h}_{\textnormal{fb}}}

\newcommand{\y}{\bm{y}}

\usepackage[noabbrev]{cleveref} 
\crefname{rem}{Remark}{Remarks}
\crefname{exam}{Example}{Examples}
\crefname{assum}{Assumption}{Assumptions}
\crefname{prop}{Proposition}{Propositions}
\crefname{propy}{Property}{Properties}
\crefname{cor}{Corollary}{Corollaries}
\crefname{lem}{Lemma}{Lemmas}
\crefname{section}{Section}{Sections}
\crefname{thm}{Theorem}{Theorems}
\crefname{alg}{Algorithm}{Algorithms}
\crefname{defn}{Definition}{Definitions}
\crefname{figure}{Fig.}{Fig.}
\Crefname{figure}{Figure}{Figures}
\crefname{equation}{}{}

\IEEEoverridecommandlockouts                              

\title{ \bf
Aggressiveness-Aware Learning-based Control of Quadrotor UAVs with Safety Guarantees
}

\author{Leonardo Colombo, Thomas Beckers, Juan Giribet
\thanks{L. Colombo is with Centre for Automation and Robotics (CSIC-UPM), Ctra. M300 Campo Real, Km 0,200, Arganda
del Rey - 28500 Madrid, Spain.{\tt\small leonardo.colombo@csic.es}} \thanks{Thomas Beckers is with the Department of Computer Science, Vanderbilt University, Nashville, TN 37212, USA {\tt\small thomas.beckers@vanderbilt.edu}} \thanks{Juan I. Giribet is with Universidad de San Andr\'es (UdeSA) and CONICET, Argentina.
        {\tt\footnotesize jgiribet@conicet.gov.ar}} \thanks{ L. Colombo acknowledge financial support from Grant PID2022-137909NB-C21 funded by MCIN/AEI/ 10.13039/501100011033. The research leading to these results was supported in part by iRoboCity2030-CM, Robótica Inteligente para Ciudades Sostenibles (TEC-2024/TEC-62), funded by the Programas de Actividades I+D en Tecnologías en la Comunidad de Madrid. J. Giribet was supported by PICT-2019-2371 and PICT-2019-0373 projects from Agencia Nacional de Investigaciones Cient\'ificas y Tecnol\'ogicas, and UBACyT-0421BA project from the Universidad de Buenos Aires (UBA), Argentina.}
}

\begin{document}

\maketitle
\thispagestyle{empty}
\pagestyle{empty}

\begin{abstract}
This paper presents an aggressiveness-aware control framework for quadrotor UAVs that integrates learning-based oracles to mitigate the effects of unknown disturbances. Starting from a nominal tracking controller on $\mathrm{SE}(3)$, unmodeled generalized forces and moments are estimated using a learning-based oracle and compensated in the control inputs. An aggressiveness-aware gain scheduling mechanism adapts the feedback gains based on probabilistic model-error bounds, enabling reduced feedback-induced aggressiveness while guaranteeing a prescribed practical exponential tracking performance. The proposed approach makes explicit the trade-off between model accuracy, robustness, and control aggressiveness, and provides a principled way to exploit learning for safer and less aggressive quadrotor maneuvers.
\end{abstract}

\section{Introduction}

Quadrotor UAVs are routinely required to track aggressive trajectories in cluttered and dynamic environments, where fast response and precise motion are key to safe operation.
For this purpose, geometric control on $SE(3)$ has emerged as a principled approach that avoids singularities and yields Lyapunov-based stability guarantees  \cite{lee2010geometric}.
Aggressive flight and trajectory-generation pipelines further highlight the need for high-bandwidth tracking in practice \cite{MellingerKumar2011MinSnap,MellingerKumar2012IJRR, lu2020robust, rodriguez2022autonomous, falanga2017aggressive, lopez2017aggressive}.

In real flight, quadrotors are affected by complex aerodynamic phenomena (drag, blade flapping, ground effect), rotor and frame asymmetries, and exogenous disturbances such as wind gusts.
Such effects can be naturally represented as unknown generalized forces and moments acting on the rigid-body dynamics.
When these unknown terms are not accounted for, nominal geometric controllers may exhibit degraded tracking and steady-state errors.
A widely used workaround is to increase feedback gains to recover performance.
While effective from a robustness standpoint, high-gain feedback typically increases \emph{feedback-induced aggressiveness}, manifesting as large variations and high rates in thrust and body torques, increased energy consumption, actuator saturation, and potentially unsafe behavior in proximity to humans or obstacles \cite{RashadBicegoZultSanchezEscalonillaJiaoFranchiStramigioli2022EnergyAwareImpedance, 8291488, BreedenGargPanagou2022CBFSampledData}.
This motivates making explicit the trade-off between tracking performance and feedback-induced aggressiveness, and developing design tools that achieve a prescribed tracking tolerance with the least aggressive feedback compatible with uncertainty.

This paper proposes to reduce the reliance on high feedback gains by augmenting geometric tracking control with a learned-based model of the unknown generalized disturbance forces and moments.
The key idea is simple: better disturbance compensation reduces model mismatch, which in turn allows smaller feedback gains while maintaining the same tracking specification.
To quantify this effect, we introduce a feedback-induced aggressiveness measure based on the sensitivity of the commanded generalized forces with respect to tracking-error variations.
Inspired by stiffness-alteration viewpoints in feedback design \cite{DellaSantina2017SoftRobots, beckers2022learning}, we formalize aggressiveness as a local input-sensitivity metric of the feedback term.
We then cast gain selection as an optimization problem: minimize feedback-induced aggressiveness subject to a practical exponential tracking bound with a prescribed ultimate error level.
To enable implementable scheduling, we leverage high-probability model-error bounds provided by Gaussian-process (GP) oracles \cite{Beckers2019StableGP, 11312152, Berkenkamp2015ECC,Berkenkamp2015SafeOpt}. GP oracles for quadrotors UAVs have been studied in \cite{nieto2026dual, nieto2024safe, 11320442, beckers2022safe, beckers2021online, colombo2023learning}.

Contribution: We formulate an \emph{aggressiveness-aware} performance objective for geometric quadrotor tracking control, capturing gain-induced sensitivity of thrust and torques via a local input-sensitivity metric. A learning-augmented controller is proposed that combines nominal geometric control with  learning-based models of unknown generalized forces and moments, and characterize how residual model error affects the feedback gains required to meet a prescribed exponential tracking specification.
We pose gain tuning as minimizing feedback-induced aggressiveness subject to a desired tracking-error bound, yielding a principled target for safe and accurate tracking in the presence of unmodeled dynamics, and enabling an aggressiveness-aware gain scheduling mechanism based on probabilistic error bounds.

The remainder of the paper is structured as follows.
Section~\ref{sec:def} reviews the quadrotor model, the nominal geometric tracking controller, and introduces the aggressiveness measure.
Section~\ref{sec:framework} presents the learning-augmented controller and the aggressiveness-aware stability and scheduling results.
Section~\ref{sec:sim} concludes the paper with numerical results, discussions, and highlights future research.

\section{System Modeling and Nominal Control}\label{sec:def}
In this section, we introduce the modeling and control of quadrotor UAVs, followed by the problem setting. In all the paper the mapping $(\cdot)^\wedge:\R^3\rightarrow\mathfrak{so}(3)$ and its inverse $(\cdot)^\vee$ denote the standard hat and vee isomorphisms between vectors in $\R^3$ and elements of the Lie algebra $\mathfrak{so}(3)$ of $3\times 3$ skew symmetric matrices.
\subsection{System class}\label{sec:sc}
We consider the  model of a quadrotor UAV evolving on the special Euclidean group $SE(3)$ as in~\cite{lee2010geometric}.
The translational and rotational dynamics of the quadrotor are given by
\begin{align}\label{for:se3system}
\begin{split}
    \dot{\bm{p}} &= \bm{v},\quad \dot{R} = R\hat{\bm{\omega}}\\
    m\dot{\bm{v}} &= m g \bm{e}_3 - T R\bm{e}_3 + \f_{\text{trans}}(\x),\\
        J\dot{\bm{\omega}} + \bm{\omega}\times J\bm{\omega} &= \bm{\tau}_b + \f_{\text{rot}}(\x),
\end{split}
\end{align}
where $\bm{p},\bm{v}\in\R^3$ denote the position and velocity of the center of mass in the inertial frame, $R\in SO(3)$ is the rotation matrix from body to inertial frame, $\bm{\omega}\in\R^3$ is the angular velocity expressed in the body frame, $m>0$ is the mass, $J\in\R^{3\times 3}$ is the positive definite inertia matrix, $g>0$ is the gravitational constant, and $\bm{e}_3 = [0,0,1]^\top$. The control input is
$\displaystyle{\bm{u} = \begin{bmatrix} T \\ \bm{\tau}_b \end{bmatrix} \in \mathcal{U} \subset \R^4}$, where $T\in\R$ denotes the total thrust magnitude and $\bm{\tau}_b\in\R^3$ are the control torques expressed in the body frame. The admissible input set $\mathcal{U}$ collects physical constraints such as thrust and torque saturation limits. The full system state is
\begin{align}
    x \coloneqq (\bm{p}, \bm{v}, R, \bm{\omega}) \in \mathcal{X} := \R^3 \times \R^3 \times SO(3) \times \R^3,
\end{align}
where $(\bm p,R)\in SE(3)$ are the configuration variables (position and attitude), and $(\bm v,\bm\omega)\in\R^3\times\R^3$ are the corresponding linear and angular velocities. During the analysis we assume that the commanded inputs remain within the admissible set $\mathcal U$ on the compact set $\mathcal X_c$.

The terms $\f_{\text{trans}}:\mathcal{X}\to\R^3$ and $\f_{\text{rot}}:\mathcal{X}\to\R^3$ represent unknown disturbances and unmodeled dynamics acting on the translational and rotational subsystems, respectively (e.g., aerodynamic drag, wind gusts, and rotor asymmetries). We also define the stacked generalized disturbance forces
\begin{align}
    \f(\x) \coloneqq 
    \begin{bmatrix}
        \f_{\text{trans}}(\x)\\
        \f_{\text{rot}}(\x)
    \end{bmatrix}
    \in \R^6.
\end{align}

Using local coordinates induced by the Lie algebra $\mathfrak{so}(3)$ via the hat and vee isomorphisms, the system dynamics~\eqref{for:se3system} can be represented in local Euclidean coordinates as a control-affine system of dimension $12$,
\begin{align}\label{for:se3affine}
    \dot{\x} = \f_{\text{dyn}}(\x) + G_{\text{dyn}}(\x)\bm{u} + \bar{\f}(\x),
\end{align}
where $\x\in\R^{12}$ denotes the vector of local coordinates associated with $(\bm{p},\bm{v},R,\bm{\omega}) \in \mathcal X$.
Here, $\f_{\text{dyn}}:\R^{12}\to\R^{12}$ collects the known part of the quadrotor dynamics, $G_{\text{dyn}}:\R^{12}\to\R^{12\times 4}$ is the known input matrix. The term $\bar{\f}(\x)\in\R^{12}$ denotes the lifting of the generalized disturbance forces $\f(\x)$ into the state-space coordinates, i.e., the contribution of $\f_{\text{trans}}$ and $\f_{\text{rot}}$ to the time derivatives of $(\bm{v},\bm{\omega})$ in the local Euclidean representation. We consider matched additive disturbances in the generalized coordinates, i.e., $\bar{\f}(\x)$ affects only the acceleration components in the local state representation. The input matrix $G_{\text{dyn}}(\x)$ has constant rank $4$ for all $\x\in\mathcal X$, reflecting the underactuated nature of the quadrotor system. 

We are interested in tracking a sufficiently smooth reference trajectory evolving on $\mathcal{X}$,
\begin{align}
    (\bm{p}_d(t), \dot{\bm{p}}_d(t), R_d(t), \bm{\omega}_d(t)) \in \mathcal X,
\end{align}
where $\bm{p}_d(t)\in\R^3$ denotes the desired position, $R_d(t)\in SO(3)$ the desired attitude, and $\dot{\bm{p}}_d(t)$ and $\bm{\omega}_d(t)$ the corresponding desired linear and angular velocities.

The associated tracking errors are defined as
\begin{align}
    \bm{e}_p &= \bm{p}-\bm{p}_d, \quad 
    \bm{e}_v = \bm{v}-\dot{\bm{p}}_d, \\
    \bm{e}_R &= \tfrac{1}{2}\big(R_d^\top R - R^\top R_d\big)^\vee, \quad
    \bm{e}_\omega = \bm{\omega} - R^\top R_d\,\bm{\omega}_d,
\end{align}
and collected in the stacked error vector
\begin{align}
    \bm{e}(t) =
    \begin{bmatrix}
        \bm{e}_p^\top &
        \bm{e}_v^\top &
        \bm{e}_R^\top &
        \bm{e}_\omega^\top
    \end{bmatrix}^\top
    \in \R^{12}.\label{e}
\end{align}

\begin{rem}
The nominal geometric tracking controller on $SE(3)$ renders the tracking error almost-globally exponentially stable, except for a measure-zero set due to the topology of $SO(3)$~\cite{lee2010geometric}. 
In this work, we perform the analysis in local coordinates induced by the Lie algebra $\mathfrak{so}(3)$, which is sufficient for the study of learning-based disturbance compensation and feedback-induced aggressiveness. All results therefore hold locally around the desired trajectory, while inheriting the almost-global properties of the underlying geometric controller.\hfill$\diamond$
\end{rem}

We assume the existence of a \emph{nominal} geometric tracking controller of the form
\begin{align}\label{for:se3basicctrl}
    \bm{u} = \hff(\x,\x_d) + \hfb(\x)H\bm{e},
\end{align}
where $\hff:\mathcal X\times\mathcal X\to\R^4$ is a feedforward term derived from the known dynamics $\f_{\text{dyn}}$ and $G_{\text{dyn}}$ (e.g., as in geometric tracking control on $SE(3)$ given in ~\cite{lee2010geometric}), $\hfb:\mathcal X\to\R^{4\times 12}$ is a state-dependent feedback mapping, and $H\in\mathcal H\subseteq\R^{12\times 12}$ is a matrix of tunable feedback gains, typically chosen to be diagonal or block-diagonal and collecting the proportional and derivative gains associated with translational and rotational error components.

The nominal control law~\eqref{for:se3basicctrl} is designed under the assumption $\bar{\f}\equiv 0$, i.e., neglecting the unknown disturbances and unmodeled dynamics. 
Under this assumption, the resulting \emph{nominal} closed-loop state dynamics is
\begin{align}\label{for:se3nomcl}
    \dot{\x} = \f_{\text{dyn}}(\x) + G_{\text{dyn}}(\x)\big[\hff(\x,\x_d)+\hfb(\x)H\bm{e}\big].
\end{align}
When expressed in terms of the tracking error variables $\bm e=(\bm e_p,\bm e_v,\bm e_R,\bm e_\omega)$, this dynamics induce a closed-loop error system of the form
\begin{align}
    \dot{\bm{e}} = \bar{\f}_{\text{cl}}(\bm{e}, \x_d, H),
\end{align}
where $\bar{\f}_{\text{cl}}:\R^{12}\times\mathcal X\times\mathcal H\to\R^{12}$ denotes the \emph{nominal closed-loop error vector field}, obtained by substituting~\eqref{for:se3nomcl} into the time derivative of the error coordinates and expressing the result in local coordinates induced by the Lie algebra $\mathfrak{so}(3)$. We restrict attention to a local chart induced by $\mathfrak{so}(3)$ along the reference such that the error map $\Psi(\x,\x_d)$ is smooth and its Jacobian w.r.t. $\x$ is bounded on the compact set $\mathcal X_c$ containing the closed-loop trajectories.

\begin{assum}\label{ass:nominal_lyap}
For each $H\in\mathcal H$, the nominal closed-loop error dynamics (i.e., $\bar{\f}\equiv 0$) admit a continuously differentiable Lyapunov function
$V(\bm e)=\bm e^\top P(H)\bm e$ with $P(H)=P(H)^\top\succ 0$ and constants $\underline\lambda,\overline\lambda,c_1>0$ such that, for all $\bm e$ in a neighborhood $\mathcal E\subset\R^{12}$,
\begin{align}\label{eq:nom_lyap_bounds}
\underline\lambda \|\bm e\|^2 \le V(\bm e)\le \overline\lambda \|\bm e\|^2,
\qquad
\dot V(\bm e)\le -c_1 \|\bm e\|^2 .
\end{align}
\end{assum}

\begin{rem}
Assumption~\ref{ass:nominal_lyap} is satisfied whenever a geometric tracking controller on $SE(3)$ is available for the known quadrotor dynamics that renders the tracking error locally exponentially stable (see~\cite{lee2010geometric} for controller construction).\hfill$\diamond$\end{rem}

\begin{assum}\label{ass:gradV}
There exists a constant $c_2>0$ such that
\begin{align}\label{eq:gradV_bound}
\bigl\|\nabla_{\bm e} V(\bm e)\bigr\|\le c_2 \|\bm e\|,
\qquad \forall \bm e\in\mathcal E.
\end{align}
\end{assum}

\begin{rem}
For the quadratic Lyapunov function $V(\bm e)=\bm e^\top P(H)\bm e$, a suitable choice is $c_2=2\|P(H)\|$.\hfill$\diamond$
\end{rem}

\begin{rem}
In practice, the ``known'' dynamics $\f_{\text{dyn}}$ can be obtained from a simplified rigid-body model of the quadrotor with nominal parameters, while complex aerodynamic effects, unmodeled dynamics, and parametric uncertainties are absorbed into the disturbance term $\bar{\f}$. This separation simplifies the design of the nominal controller~\eqref{for:se3basicctrl}.\hfill$\diamond$
\end{rem}

\begin{rem}
Although the quadrotor is underactuated since $G_{\text{dyn}}(\x)\in\mathbb{R}^{12\times 4}$ has rank $4<12$, the formulation on $SE(3)$ together with a suitable choice of error coordinates $(\bm{e}_p,\bm{e}_v,\bm{e}_R,\bm{e}_\omega)$ allows the design of nominal tracking controllers that satisfy Assumption~\ref{ass:nominal_lyap} directly on the configuration manifold. In particular, stability can be established without introducing additional virtual states or artificial control inputs, as the coupling between translational and rotational dynamics is handled intrinsically through the geometry of $SE(3)$ and the attitude-dependent thrust direction.\hfill$\diamond$
\end{rem}

The application of the nominal control law~\eqref{for:se3basicctrl} to the \emph{actual} dynamics~\eqref{for:se3affine} generally leads to a mismatch between the theoretical performance guarantees and the observed behavior, due to the presence of the unknown disturbance term $\bar{\f}$. A naive remedy is to increase the feedback gains collected in $H$ to reduce the tracking error. However, high feedback gains are often undesirable, as they result in aggressive attitude maneuvers, large variations in thrust and body-frame torques, increased actuator stress, and potentially unsafe behavior.

To quantify this effect, we introduce a measure of \emph{feedback-induced aggressiveness} based on the sensitivity of the commanded control inputs with respect to variations in the tracking error. Since the generalized forces acting on the quadrotor are directly given by the control input $\bm{u}=(T,\bm{\tau}_b)\in\mathcal U$, we focus on the feedback component of the control law~\eqref{for:se3basicctrl}. Inspired by the stiffness-alteration viewpoint in soft robotics~\cite{DellaSantina2017SoftRobots}, we define the feedback-induced aggressiveness as the mapping $s:\mathcal H\times\mathcal X \to \R{\ge 0}$
\begin{align}
    s(H,\x) \coloneqq
    \left\|
        \frac{\partial}{\partial \bm{e}}
        \big( \hfb(\x) H \bm{e} \big)
    \right\|=\left\|\hfb(\x) H \right\|,
    \label{for:se3feedbackstiff}
\end{align}
where the derivative is taken with respect to the tracking error $\bm{e}\in\R^{12}$, treating the state $\x$ as a parameter, and the norm denotes an induced matrix norm. The quantity $s(H,\x)$ captures how strongly the commanded thrust and torques react to small deviations in the tracking error. Large values of $s(H,\x)$ correspond to aggressive feedback behavior, which can be problematic in proximity to obstacles or humans and may also excite unmodeled dynamics.

\subsection{Motivating Example}\label{sec:se3mot}

To illustrate the trade-off between feedback gains, tracking performance, and feedback-induced aggressiveness in the quadrotor setting, we consider a simplified model of vertical motion around hover. 
This model is obtained from~\eqref{for:se3system} by assuming small attitude deviations and focusing on the altitude dynamics. 
Let $x_1$ denote the altitude deviation from a desired constant height, and $x_2$ the corresponding vertical velocity. 
The dynamics are given by
\begin{align}
   \dot{x}_1 &= x_2,\\
   \dot{x}_2 &= \underbrace{g - \tfrac{1}{m}T}_{\f_{\text{dyn}}(x)} + \underbrace{f(x)}_{\text{unknown}},
\end{align}
where $m>0$ denotes the quadrotor mass, $T$ is the total thrust command, and $f(x)$ collects unknown effects such as ground effect, unmodeled aerodynamic drag, and external disturbances. 
The control objective is to regulate the system to the equilibrium $x_d=\bm{0}$.

A nominal controller for the vertical channel can be obtained by treating the thrust $T$ as a virtual input and canceling the known gravity term,
\begin{align}
    T = m\big(g + H\bm e\big)
    = m\big(g + h_2 x_2 + h_1 x_1\big),\label{for:se3ex:ctrl}
\end{align} where $\bm{e} = [x_1,x_2]^\top$ and $H=[h_1,h_2]\in\mathcal{H}=\R_{\geq 0}^{1\times 2}$ collects the feedback gains. 
The resulting closed-loop dynamics is
\begin{align}
    \dot{x}_2 = -h_2 x_2 - h_1 x_1 + f(x).\label{for:se3clexamp}
\end{align}

From a purely nominal control perspective and neglecting the unknown term $f(x)$, increasing $h_1$ and $h_2$ improves the rate of convergence of the tracking error. 
In this case the feedback-induced aggressiveness becomes
\begin{align}
    s(H,\x)
    = \left\Vert \frac{\partial}{\partial \bm{e}} \big( m H \bm{e} \big) \right\Vert
    = m\|H\|_2
    = m\sqrt{h_1^2 + h_2^2}.
\end{align}

In the context of a quadrotor, increased aggressiveness manifests as large thrust variations around hover, leading to aggressive vertical accelerations, higher energy consumption, and potential actuator saturation.

The required feedback gains can be reduced by incorporating a learned model $\hat{f}$ of the unknown term $f$. 
The control law~\eqref{for:se3ex:ctrl} is then modified to
\begin{align}
    T = m\big(g + h_2 x_2 + h_1 x_1\big) - m\hat{f}(x),
\end{align}
which yields the closed-loop dynamics
\begin{align}
    \dot{x}_2 = -h_2 x_2 - h_1 x_1 + \big[f(x)-\hat{f}(x)\big].
\end{align}
In the ideal case of perfect model matching, i.e., $f=\hat{f}$, the equilibrium $x_d=\bm{0}$ is asymptotically stable for any choice of $h_1,h_2>0$. 
More generally, improved model accuracy allows one to reduce the feedback gains required to achieve a prescribed tracking performance, thereby lowering the feedback-induced aggressiveness and avoiding overly aggressive thrust commands. 
Conversely, larger model errors necessitate higher feedback gains, resulting in more aggressive and potentially less safe behavior.

\subsection{Problem Setting}\label{sec:se3ps}

The motivating example above illustrates the fundamental trade-off between feedback gains, tracking performance, and feedback-induced aggressiveness in a simplified vertical quadrotor model. An analogous trade-off arises for the full dynamics~\eqref{for:se3affine}: increasing the feedback gains collected in $H$ typically improves disturbance attenuation under the assumed continuous-time model and in the absence of input constraints, but simultaneously amplifies the sensitivity of the commanded thrust and body-frame torques with respect to tracking-error variations, as quantified by the  measure~\eqref{for:se3feedbackstiff}.

The objective of this article is to guarantee a prescribed tracking performance for the quadrotor while minimizing the feedback-induced aggressiveness associated with the control law~\eqref{for:se3basicctrl}. To this end, we consider the system~\eqref{for:se3affine}, where the unknown generalized disturbance forces are approximated by a learning-based oracle $\hat{\f}:\mathcal X \to \R^{6}$. The corresponding compensation is mapped into the physical input space via $K_{\text{dyn}}(\x)$ and applied through the input matrix $G_{\text{dyn}}(\x)$. We adopt the augmented control law
\begin{align}
    \bm{u} = \hff(\x,\x_d) - K_{\text{dyn}}(\x)\hat{\f}(\x) + \hfb(\x)H\bm{e},\label{for:se3augctrl}
\end{align}
where $K_{\text{dyn}}(\x)\in\R^{4\times 6}$ is a known disturbance-compensation map that distributes the learned generalized forces and moments into the physical input space. Typical choices for $K_{\text{dyn}}(\x)$ include pseudoinverse- or least-squares-based mappings associated with $G_{\text{dyn}}(\x)$; however, perfect cancellation of $\bar{\f}$ is not assumed. The functions $\hff$ and $\hfb$ are as defined in~\eqref{for:se3basicctrl}.

The resulting closed-loop tracking error dynamics depend on the (state-space) residual mismatch $\Delta(\x)=\bar{\f}(\x)-G_{\text{dyn}}(\x)K_{\text{dyn}}(\x)\hat{\f}(\x)$ and on the feedback gain matrix $H$.
We seek to select $H\in\mathcal H$ so that the tracking error remains uniformly bounded and converges \emph{practically exponentially} to a neighborhood of the origin for all initial conditions in a (local) region of attraction.  Specifically, for a prescribed ultimate tracking-error bound $\varepsilon\in\R_{>0}$, we require that for each feasible $H\in\mathcal H$ there exist constants $\gamma_1(H),\gamma_2(H)\in\R_{>0}$ such that the closed-loop tracking error satisfies: \begin{align}
\begin{split}\label{for:se3maxtrack}
        &\min_{H\in\mathcal{H}} \; \sup_{\x\in\mathcal X_c} s(H,\x) \\
        &\text{s.t. } 
        \Vert\bm{e}(t)\Vert \leq 
        \gamma_1 e^{-\gamma_2 (t-t_0)} \Vert \bm{e}(t_0)\Vert + \varepsilon,
        \quad \forall t \ge t_0,
\end{split}
\end{align}
where $\mathcal X_c\subset\mathcal X$ is a compact set containing the relevant closed-loop trajectories, $s(H,\x)$ is evaluated locally for $\bm e\approx 0$, and the norm denotes any (convex) matrix norm (e.g., induced 2-norm or Frobenius norm).

\begin{rem}
Since the feedback term is linear in $\bm e$, we have $s(H,\x)=\|\hfb(\x)H\|$, hence the objective $H\mapsto \sup_{\x\in\mathcal X_c} s(H,\x)$ is convex. If, in addition, the feasible set is nonempty and $\mathcal H$ is compact (e.g., by imposing $\|H\|\le \bar H$), then problem~\eqref{for:se3maxtrack} admits at least one minimizer.\hfill$\diamond$
\end{rem}

\section{Aggressiveness-Aware Learning Control}\label{sec:framework}
We address the problem in~\eqref{for:se3maxtrack} with a learning-augmented control framework that balances model-based disturbance compensation and feedback gains such that a prescribed tracking performance is achieved while reducing feedback-induced aggressiveness.

The controller implements the augmented control law
\begin{align}\label{for:framework:ctrl}
    \bm{u} = \hff(\x,\x_d) - K_{\text{dyn}}(\x)\hat{\f}_N(\x) + \hfb(\x)H_N\bm{e},
\end{align}
which augments the nominal geometric controller with a disturbance compensation term and an aggressiveness-aware feedback gain matrix $H_N$. 
Here, $\hat{\f}_N(\x)\in\R^6$ denotes the oracle prediction of the generalized disturbance forces and moments, and $K_{\text{dyn}}(\x)\in\R^{4\times 6}$ maps these estimates into the input space. 
When the oracle admits a uniform high-probability error bound on $\mathcal X_c$, smaller gains in $H_N$ can be sufficient to satisfy the prescribed practical tracking bound, thereby reducing the feedback-induced aggressiveness. Conversely, larger model errors require increased feedback gains to maintain robustness.

\subsection{Oracle and Model Error Bounds}\label{sec:oracle}
Recall that the unknown disturbances in~\eqref{for:se3system} are captured by the generalized disturbance forces and moments
\begin{align}
    \f(\x) =
    \begin{bmatrix}
        \f_{\text{trans}}(\x)\\
        \f_{\text{rot}}(\x)
    \end{bmatrix}
    \in\R^6.
\end{align}
We consider an oracle that predicts $\f(\x)$ for a given state $\x\in\mathcal X$. 
The oracle is trained from a dataset constructed from measured or estimated disturbance realizations obtained during closed-loop operation.

Specifically, we define a time-varying dataset
\begin{align}
    \mathcal{D}(t) = \{(\x^{\{i\}}, \y^{\{i\}})\}_{i=1}^{N(t)},\label{for:dataset}
\end{align}
where $N:\R_{\ge 0}\to\N$ denotes the (time-varying) number of data points and $\y^{\{i\}}\in\R^6$ is a data-driven estimate of $\f(\x^{\{i\}})$. 
In practice, $\y$ can be obtained by subtracting the known model contributions from measured accelerations and applied inputs in~\eqref{for:se3system}.

For notational simplicity, we consider piecewise-constant datasets over intervals $t\in[t_N,t_{N+1})$ with $0=t_0<t_1<t_2<\cdots$ and denote $\mathcal D_N:=\mathcal D(t)$ on $[t_N,t_{N+1})$. 
The corresponding oracle at time $t_N$ is denoted by
\begin{align}
    \hat{\f}_N:\mathcal X \to \R^6,
    \qquad
    \hat{\f}_N(\x)=\hat{\y}\,\vert\,\x,\mathcal D_N.\label{for:oraclepred}
\end{align}
This formulation accommodates purely online learning, offline training with fixed datasets, as well as hybrid schemes. Between update times $t_N$, the oracle and gains are held constant; the closed-loop system is thus piecewise-smooth with switching at $\{t_N\}$.

To quantify oracle accuracy, we assume a high-probability bound on the model error over a compact set. 
Let $\mathcal X_{\mathrm{c}}\subset\mathcal X$ be a compact set containing the relevant closed-loop trajectories.
\begin{assum}\label{ass:errorbound}
Consider an oracle with predictions $\hat{\f}_N\in\mathcal{C}^0$ based on  $\mathcal D_N$. 
There exists a finite function $\bar{\rho}_N:\mathcal X_{\mathrm{c}}\times(0,1]\to\R_{\ge 0}$ such that, for any confidence level $\delta\in(0,1]$,
\begin{align}\label{for:moderror}
    \Prob\!\left\{\ \|\f(\x)-\hat{\f}_N(\x)\|\le \bar{\rho}_N(\x,\delta)\ \ \forall \x\in\mathcal X_{\mathrm{c}}\ \right\}\ge \delta,
\end{align}
for any fixed $N\in\{0,1,\ldots,N_{\text{end}}\}$.
\end{assum}

\begin{assum}\label{ass:switching}
The number of dataset updates is finite: there exist $N_{\text{end}}\in\N$ and $T_{\text{end}}\in\R_{\ge 0}$ such that $\mathcal D(t)=\mathcal D_{N_{\text{end}}}$ for all $t\ge T_{\text{end}}$.
\end{assum}

\begin{rem}The probability is taken w.r.t  the measurement noise and the resulting oracle posterior (and, if applicable, the prior over $f$). Assumption~\ref{ass:errorbound} ensures bounded model error in probability on $\mathcal X_{\mathrm{c}}$, while Assumption~\ref{ass:switching} reflects practical limitations in memory and computational resources.\hfill$\diamond$\end{rem}

\subsection{Gaussian Process as Oracle}\label{sec:gp}
Gaussian process (GP) models provide a principled oracle for nonlinear regression with uncertainty quantification. 
We adopt an independent-output GP construction to predict each component of $\f(\x)\in\R^6$. Let $\mathcal D_N$ contain $N_{\mathcal D}$ training pairs and define the input matrix $X=[\x^{\{1\}},\x^{\{2\}},\ldots,\x^{\{N_{\mathcal D}\}}]\in\R^{12\times N_{\mathcal D}}$ and the output matrix $Y^\top=[\tilde{\y}^{\{1\}},\tilde{\y}^{\{2\}},\ldots,\tilde{\y}^{\{N_{\mathcal D}\}}]\in\R^{6\times N_{\mathcal D}}$, where $\tilde{\y}=\y+\bm\eta$ are noisy measurements corrupted by Gaussian noise $\bm\eta\sim\mathcal{N}(0,\sigma^2 I_6)$. 
For a query $\x^*$, the GP posterior mean and variance for each output component $i\in\{1,\ldots,6\}$ are
\begin{align}
	\mu_i(\x^*\mid \mathcal D_N) &= \bm{k}(\x^*,X)^\top K^{-1}Y_{:,i},\label{for:gppred}\\
	\sigma^2_i(\x^*\mid\mathcal D_N) &= k(\x^*,\x^*)-\bm{k}(\x^*,X)^\top K^{-1} \bm{k}(\x^*,X),\notag
\end{align}
where $k$ is a positive definite kernel, $\bm{k}(\x^*,X)=[k(\x^*,\x^{\{1\}}),\ldots,k(\x^*,\x^{\{N_{\mathcal D}\}})]^\top$, and $K\in\R^{N_{\mathcal D}\times N_{\mathcal D}}$ is the corresponding Gram matrix.

Stacking the component-wise posteriors yields a Gaussian distribution with mean $\bm\mu(\x^*\mid\mathcal D_N)\in\R^6$ and diagonal covariance $\Sigma(\x^*\mid\mathcal D_N)\in\R^{6\times 6}$. 
Under standard regularity assumptions on the unknown function, this representation enables explicit high-probability error bounds.

\begin{assum}\label{ass:rkhs}
The kernel $k$ is chosen such that each component $f_i$ of $\f$ has a finite reproducing kernel Hilbert space norm on $\mathcal X_{\mathrm{c}}$, i.e., $\|f_i\|_{k}\le B_i <\infty$ for all $i=1,\ldots,6$.
\end{assum}

Under Assumption~\ref{ass:rkhs} and standard GP concentration results, the prediction error admits a computable high-probability bound of the form~\eqref{for:moderror}, where $\bar{\rho}_N(\x,\delta)$ can be expressed in terms of the posterior variance and information-gain quantities (see, e.g., \cite{Srinivas2010GPUCB}).

\subsection{Aggressiveness-Aware Stability Guarantees}
\begin{thm}\label{thm:main}
Consider the system~\eqref{for:se3affine} together with the learning-based control law~\eqref{for:framework:ctrl}. 
Suppose that Assumptions~\ref{ass:nominal_lyap} and~\ref{ass:gradV} hold, and that the oracle satisfies Assumption~\ref{ass:errorbound} on a compact set $\mathcal X_c\subset\mathcal X$ containing the closed-loop trajectories.

Fix $\varepsilon>0$ and $\delta\in(0,1]$. 
Assume the oracle error bound satisfies
\begin{align}\label{eq:gain_condition}
    \sup_{\x\in\mathcal X_c}\bar\rho_N(\x,\delta)\ \le\ \frac{c_1}{2c_2}\,\varepsilon,
\end{align}
where $c_1,c_2$ correspond to the Lyapunov inequalities for the chosen $H_N$. 
Denote by $\mathcal E\subset\R^{12}$ a neighborhood of the origin in which the local error coordinates are valid and the nominal Lyapunov conditions hold. 
Then, with probability at least $\delta$, every solution with $\x(t)\in\mathcal X_c$ and $\bm e(t_0)\in\mathcal E$ satisfies the practical exponential tracking bound
\begin{align}\label{eq:pexp_bound}
    \|\bm e(t)\|\le \gamma_1 e^{-\gamma_2(t-t_0)}\|\bm e(t_0)\|+\varepsilon,\qquad \forall t\ge t_0,
\end{align}
for explicit constants $\gamma_1(H_N),\gamma_2(H_N)>0$.

Moreover, assume that the gain-selection rule produces $H_N$ such that, on $\mathcal X_c$,
\begin{align}\label{eq:HN_affine_bound}
    \|H_N\|\le \kappa_1 \sup_{\x\in\mathcal X_c}\bar\rho_N(\x,\delta)+\kappa_2
\end{align}
for some constants $\kappa_1,\kappa_2>0$ independent of $N$. 
Then the feedback-induced aggressiveness satisfies, for all $\x\in\mathcal X_c$,
\begin{align}\label{eq:aggr_bound_thm}
s(H_N,\x)=\|\hfb(\x)H_N\|
\le \alpha_1\,\bar\rho_N(\x,\delta)+\alpha_2,
\end{align}
for constants $\alpha_1,\alpha_2>0$ independent of $N$. 
In particular, smaller model-error bounds $\bar\rho_N$ enable lower feedback-induced aggressiveness while preserving~\eqref{eq:pexp_bound}.
\end{thm}

\medskip

\textit{Proof:} Under the control law~\eqref{for:framework:ctrl}, the closed-loop dynamics is
\begin{align*}
\dot{\x} =& \f_{\mathrm{dyn}}(\x) + G_{\mathrm{dyn}}(\x)\Bigl(\hff(\x,\x_d)+\hfb(\x)H_N \bm e\Bigr)\\&
\;+\;\bar{\f}(\x)\;-\;G_{\mathrm{dyn}}(\x)K_{\mathrm{dyn}}(\x)\hat{\f}_N(\x).
\end{align*}
Let $\x_d(t)$ be the reference and let $\bm e=\Psi(\x,\x_d)\in\R^{12}$ denote the stacked tracking error map
$(\bm e_p,\bm e_v,\bm e_R,\bm e_\omega)$. By construction $\Psi(\cdot,\x_d)$ is smooth in local coordinates on the considered chart.
Differentiating $\bm e=\Psi(\x,\x_d)$ yields
\[
\dot{\bm e} = \frac{\partial \Psi}{\partial \x}(\x,\x_d)\dot{\x} + \frac{\partial \Psi}{\partial \x_d}(\x,\x_d)\dot{\x}_d.
\]
Define the \emph{nominal} error $\bar \f_{\text{cl}}(\bm e,\x_d,H_N)$ as the right-hand side obtained by setting
$\bar{\f}\equiv 0$ and $\hat{\f}_N\equiv 0$, i.e., by replacing $\dot{\x}$ above with
$\f_{\mathrm{dyn}}(\x)+G_{\mathrm{dyn}}(\x)\bigl(\hff(\x,\x_d)+\hfb(\x)H_N\bm e\bigr)$.
Then the \emph{actual} error dynamics can be written as
\begin{align}\label{eq:err_dyn_proof}
\dot{\bm e} = \bar \f_{\text{cl}}(\bm e,\x_d,H_N) + \bm d_N(\x),
\end{align}
where the perturbation term $\bm d_N(\x)$ is
\begin{align}\label{eq:dN_def}
\frac{\partial \Psi}{\partial \x}(\x,\x_d)\Bigl(\bar{\f}(\x)-G_{\mathrm{dyn}}(\x)K_{\mathrm{dyn}}(\x)\hat{\f}_N(\x)\Bigr).
\end{align}
Since $\frac{\partial \Psi}{\partial \x}$, $G_{\mathrm{dyn}}$, and $K_{\mathrm{dyn}}$ are known smooth maps and $\x(t)\in\mathcal{X}_c$ by assumption,
there exists $c_J>0$ such that $\bigl\|\frac{\partial \Psi}{\partial \x}(\x,\x_d)\bigr\|\le c_J$ for all $\x\in\mathcal X_c$. Define
\[
\Delta_N(\x)\coloneqq \bar{\f}(\x)-G_{\mathrm{dyn}}(\x)K_{\mathrm{dyn}}(\x)\hat{\f}_N(\x)\in\R^{12}.
\]
Then \eqref{eq:dN_def} can be compactly written as $\bm d_N(\x)=\mathcal J(\x)\Delta_N(\x)$ with
$\mathcal J(\x)\coloneqq \frac{\partial \Psi}{\partial \x}(\x,\x_d)$ bounded on $\mathcal X_c$.

To connect the oracle bound in Assumption~\ref{ass:errorbound} to $\Delta_N$, recall that $\bar{\f}(\x)$ is the lifting of the generalized disturbance $\f(\x)\in~\R^6$ into local state coordinates. 
Thus, on $\mathcal X_c$ there exists a bounded map $B(\x)$ such that $\bar{\f}(\x)=B(\x)\f(\x)$. 
Moreover, since $G_{\mathrm{dyn}}$ and $K_{\mathrm{dyn}}$ are continuous on the compact set $\mathcal X_c$, there exists $c_{GK}>0$ such that $\|G_{\mathrm{dyn}}(\x)K_{\mathrm{dyn}}(\x)\|\le c_{GK}$ for all $\x\in\mathcal X_c$.
Consequently, there exists a constant $c_\Delta>0$ such that, on $\mathcal X_c$, $\|\Delta_N(\x)\|\le c_\Delta\,\|\f(\x)-\hat{\f}_N(\x)\|$.

Therefore, on the event $\mathcal A_\delta=\{\|\f(\x)-\hat{\f}_N(\x)\|\le \bar\rho_N(\x,\delta)\ \ \forall \x\in\mathcal X_c\}$,
which has probability at least $\delta$ by Assumption~\ref{ass:errorbound}, 
$\|\Delta_N(\x)\|\le c_\Delta\,\bar\rho_N(\x,\delta),\,\forall \x\in\mathcal X_c$.

Absorbing the constant factor $c_Jc_\Delta$ into $\bar\rho_N$ (i.e., redefining $\bar\rho_N$ by a constant scaling), we may assume without loss of generality that \eqref{eq:err_dyn_proof} holds with $\bm d_N(\x)=\Delta_N(\x)$ and
$\Prob\!\left\{\,\|\Delta_N(\x)\|\le \bar\rho_N(\x,\delta)\ \ \forall \x\in\mathcal X_c\,\right\}\ge \delta$.

Let $V(\bm e)=\bm e^\top P(H_N)\bm e$ be the Lyapunov function from Assumption~\ref{ass:nominal_lyap}. Along trajectories of~\eqref{eq:err_dyn_proof},
\[
\dot V(\bm e) = \nabla_{\bm e} V(\bm e)^\top \bar \f_{\mathrm{cl}}(\bm e,\x_d,H_N) + \nabla_{\bm e} V(\bm e)^\top \Delta_N(\x).
\]
By Assumption~\ref{ass:nominal_lyap},
$\nabla_{\bm e} V(\bm e)^\top \bar \f_{\mathrm{cl}}(\bm e,\x_d,H_N)\le -c_1\|\bm e\|^2$ for all $\bm e\in\mathcal{E}$.
Using Cauchy--Schwarz and Assumption~\ref{ass:gradV},
\begin{align*}
\dot V(\bm e)\le& -c_1\|\bm e\|^2 + \|\nabla_{\bm e} V(\bm e)\|\,\|\Delta_N(\x)\|\\
\le& -c_1\|\bm e\|^2 + c_2\|\bm e\|\,\|\Delta_N(\x)\|.
\end{align*}
On the event $\mathcal{A}_\delta=\{\|\Delta_N(\x)\|\le \bar\rho_N(\x,\delta)\ \forall \x\in\mathcal{X}_c\}$,
which has probability at least $\delta$, we obtain
\begin{align}\label{eq:Vdot_event}
\dot V(\bm e)\le -c_1\|\bm e\|^2 + c_2\bar\rho_N(\x,\delta)\|\bm e\|,\qquad \forall \x\in\mathcal{X}_c.
\end{align}

Assume~\eqref{eq:gain_condition}. Then for all $\x\in\mathcal{X}_c$, $\displaystyle{c_2\bar\rho_N(\x,\delta)\le \frac{c_1}{2}\,\varepsilon}$. Hence, whenever $\|\bm e\|\ge \varepsilon$, inequality~\eqref{eq:Vdot_event} implies
\[
\dot V(\bm e)\le -c_1\|\bm e\|^2 + \frac{c_1}{2}\varepsilon\|\bm e\|
\le -\frac{c_1}{2}\|\bm e\|^2.
\]
Using $V(\bm e)\ge \underline\lambda\|\bm e\|^2$ from~\eqref{eq:nom_lyap_bounds}, we obtain for $\|\bm e\|\ge\varepsilon$,
\begin{align}\label{eq:Vdot_linear}
\dot V(\bm e)\le -\frac{c_1}{2\underline\lambda}\,V(\bm e).
\end{align}
Therefore, on any time interval during which $\|\bm e(t)\|\ge\varepsilon$, Grönwall's inequality gives
\[
V(t)\le V(t_0)\exp\!\Bigl(-\frac{c_1}{2\underline\lambda}(t-t_0)\Bigr).
\]
Furthermore, on the boundary $\|\bm e\|=\varepsilon$, \eqref{eq:Vdot_event} together with \eqref{eq:gain_condition} yields
$\dot V(\bm e)\le -\frac{c_1}{2}\varepsilon^2<0$, hence the set
$\mathcal{B}_\varepsilon=\{\bm e:\|\bm e\|\le\varepsilon\}$ is forward invariant on the event $\mathcal{A}_\delta$.
Thus, trajectories enter $\mathcal{B}_\varepsilon$ in finite time and remain there.

From~\eqref{eq:nom_lyap_bounds},
$\underline\lambda\|\bm e\|^2\le V(\bm e)\le \overline\lambda\|\bm e\|^2$.
Combining this with the decay estimate for $V$ outside $\mathcal{B}_\varepsilon$ yields
\[
\|\bm e(t)\|
\le \sqrt{\frac{\overline\lambda}{\underline\lambda}}\,
\exp\!\Bigl(-\frac{c_1}{4\underline\lambda}(t-t_0)\Bigr)\,\|\bm e(t_0)\|
\quad \text{until }\|\bm e(t)\|=\varepsilon.
\]
After the first hitting time of $\mathcal{B}_\varepsilon$, invariance implies $\|\bm e(t)\|\le\varepsilon$.
A standard stitching argument then yields~\eqref{eq:pexp_bound} with
$\gamma_1=\sqrt{\overline\lambda/\underline\lambda}$ and $\gamma_2=\frac{c_1}{4\underline\lambda}$.
This holds on $\mathcal{A}_\delta$, hence with probability at least $\delta$.

\noindent By definition,
$\displaystyle{s(H_N,\x)=\|\hfb(\x)H_N\|\le \|\hfb(\x)\|\,\|H_N\|}$.
Since $\x\in\mathcal{X}_c$ and $\mathcal{X}_c$ is compact, continuity of $\hfb$ implies
$\sup_{\x\in\mathcal{X}_c}\|\hfb(\x)\|<\infty$.
Using~\eqref{eq:HN_affine_bound}, there exist constants $\alpha_1,\alpha_2>0$, independent of $N$, such that~\eqref{eq:aggr_bound_thm} holds.\hfill$\square$

\medskip

\medskip

\begin{cor}\label{cor:aggr_limit}
Under the assumptions of Theorem~\ref{thm:main}, fix a confidence level $\delta\in(0,1]$ and assume that
\begin{align}\label{eq:rho_to_zero}
\lim_{N\to\infty}\ \sup_{\x\in\mathcal{X}_c}\bar\rho_N(\x,\delta)=0 .
\end{align}
Then, for any $\varepsilon>0$ and any fixed $H\in\mathcal H$, there exists $N_\varepsilon\in\N$ such that for all $N\ge N_\varepsilon$ the gain condition~\eqref{eq:gain_condition} holds (with $c_1,c_2$ corresponding to that $H$). Consequently, for each fixed $N\ge N_\varepsilon$, with probability at least $\delta$, the tracking error satisfies
\begin{align}\label{eq:pexp_bound_cor}
\|\bm e(t)\|\le \gamma_1 e^{-\gamma_2(t-t_0)}\|\bm e(t_0)\|+\varepsilon,\qquad \forall t\ge t_0,
\end{align}
with $\gamma_1=\gamma_1(H)$ and $\gamma_2=\gamma_2(H)$, and the feedback-induced aggressiveness satisfies 
\begin{align}\label{eq:aggr_limit_cor}
\limsup_{N\to\infty}\ \sup_{\x\in\mathcal{X}_c} s(H,\x)\ \le\ \alpha_2 .
\end{align}
\end{cor}

\medskip

\textit{Proof:} Fix $\varepsilon>0$ and any fixed $H\in\mathcal H$. By~\eqref{eq:rho_to_zero}, there exists $N_\varepsilon\in\N$ such that for all $N\ge N_\varepsilon$,
$\displaystyle{\sup_{\x\in\mathcal{X}_c}\bar\rho_N(\x,\delta)\le \frac{c_1}{2c_2}\,\varepsilon}$, where $c_1,c_2$ correspond to the Lyapunov inequalities associated with the chosen fixed gain matrix $H$ in Theorem~\ref{thm:main}. Hence the gain condition~\eqref{eq:gain_condition} of Theorem~\ref{thm:main} holds for all $N\ge N_\varepsilon$, and therefore for each fixed $N\ge N_\varepsilon$ the practical exponential bound~\eqref{eq:pexp_bound_cor} follows with probability at least $\delta$.

Moreover, Theorem~\ref{thm:main} implies that for each fixed $N\ge N_\varepsilon$ (on the corresponding event of probability at least $\delta$), for all $\x\in\mathcal{X}_c$ we have
$s(H,\x)\le \alpha_1\,\bar\rho_N(\x,\delta)+\alpha_2$.
Taking $\sup_{\x\in\mathcal{X}_c}$ on both sides and then $\limsup_{N\to\infty}$, and using~\eqref{eq:rho_to_zero}, yields
\[
\limsup_{N\to\infty}\ \sup_{\x\in\mathcal{X}_c} s(H,\x)
\le
\alpha_1 \lim_{N\to\infty}\sup_{\x\in\mathcal{X}_c}\bar\rho_N(\x,\delta) + \alpha_2
=
\alpha_2,
\]
which proves the claim.\hfill$\square$

\medskip
\begin{prop}\label{prop:block_sched}
Assume the setting of Theorem~\ref{thm:main}. Suppose that the residual perturbation in the \emph{error dynamics} can be decomposed as
\[
\Delta_N(\x)=\begin{bmatrix}\Delta_{t,N}(\x)\\ \Delta_{r,N}(\x)\end{bmatrix}\in\R^{12},
\,
\Delta_{t,N}(\x),\Delta_{r,N}(\x)\in\R^{6},
\]
with corresponding high-probability bounds
\[
\|\Delta_{t,N}(\x)\|\le \bar\rho_{t,N}(\x,\delta),\quad
\|\Delta_{r,N}(\x)\|\le \bar\rho_{r,N}(\x,\delta),
\]
$\forall \x\in\mathcal{X}_c$, holding on an event of probability at least $\delta$.

\noindent Assume further that there exist constants $c_{t,1},c_{r,1}>0$ such that, along the actual closed-loop trajectories, the Lyapunov derivative satisfies
\begin{align}\label{eq:Vdot_split}
\dot V(\bm e)\le&
-\lambda_t\|\bm e_t\|^2-\lambda_r\|\bm e_r\|^2
+ c_{t,1}\|\bm e_t\|\,\|\Delta_{t,N}(\x)\|\notag\\&
+ c_{r,1}\|\bm e_r\|\,\|\Delta_{r,N}(\x)\|,
\end{align}
where $\bm e_t=[\bm e_p^\top,\bm e_v^\top]^\top\in\R^6$,  $\bm e_r=[\bm e_R^\top,\bm e_\omega^\top]^\top\in\R^6$, and $\lambda_t,\lambda_r>0$ depend monotonically on the translational gains $(K_{p,N},K_{v,N})$ and rotational gains $(K_{R,N},K_{\omega,N})$.

Choose $H_N$ block-diagonal with translational and rotational blocks scheduled according to the gain-scheduling rule, and pick the coefficients $c_p,c_v,c_R,c_\omega$ such that
\begin{align}\label{eq:lambdacond}
\lambda_t \ge& \frac{2\sqrt{2}\,c_{t,1}}{\varepsilon}\,\sup_{\x\in\mathcal X_c}\bar\rho_{t,N}(\x,\delta),\\
\lambda_r \ge& \frac{2\sqrt{2}\,c_{r,1}}{\varepsilon}\,\sup_{\x\in\mathcal X_c}\bar\rho_{r,N}(\x,\delta).\nonumber
\end{align}
Then, on the event of probability at least $\delta$,
\[
\dot V(\bm e)\le -\tfrac{1}{2}\lambda \|\bm e\|^2
\quad\text{whenever }\|\bm e\|\ge\varepsilon,
\,
\lambda:=\min\{\lambda_t,\lambda_r\},
\]
which implies practical exponential convergence of $\bm e(t)$ to an $\varepsilon$-neighborhood of the origin (cf. Theorem~\ref{thm:main}).
\end{prop}

\medskip

\textit{Proof:}
Fix $\varepsilon>0$ and consider the event on which the bounds on $\Delta_{t,N}$ and $\Delta_{r,N}$ hold for all $\x\in\mathcal{X}_c$.
From~\eqref{eq:Vdot_split} and the bounds, for any $\x\in\mathcal X_c$,
\begin{align*}
\dot V(\bm e)\le&
-\lambda_t\|\bm e_t\|^2-\lambda_r\|\bm e_r\|^2
+ c_{t,1}\|\bm e_t\|\,\bar\rho_{t,N}(\x,\delta)
\\&+ c_{r,1}\|\bm e_r\|\,\bar\rho_{r,N}(\x,\delta).
\end{align*}

Taking suprema over $\mathcal X_c$ and defining $\displaystyle{\bar\rho_{t,N}^{\max}:=\sup_{\x\in\mathcal X_c}\bar\rho_{t,N}(\x,\delta)}$ and $
\bar\rho_{r,N}^{\max}:=\sup_{\x\in\mathcal X_c}\bar\rho_{r,N}(\x,\delta)$, 
\[
\dot V(\bm e)\le
-\lambda_t\|\bm e_t\|^2-\lambda_r\|\bm e_r\|^2
+ c_{t,1}\bar\rho_{t,N}^{\max}\|\bm e_t\|
+ c_{r,1}\bar\rho_{r,N}^{\max}\|\bm e_r\|.
\]

Now assume $\|\bm e\|\ge\varepsilon$, where $\|\bm e\|^2=\|\bm e_t\|^2+\|\bm e_r\|^2$.
Since $\|\bm e_t\|+\|\bm e_r\|\le \sqrt{2}\,\|\bm e\|$, the scheduling conditions~\eqref{eq:lambdacond} imply
\[
c_{t,1}\bar\rho_{t,N}^{\max}\|\bm e_t\|
\le \frac{\lambda_t\varepsilon}{2\sqrt{2}}\,\|\bm e_t\|,
\qquad
c_{r,1}\bar\rho_{r,N}^{\max}\|\bm e_r\|
\le \frac{\lambda_r\varepsilon}{2\sqrt{2}}\,\|\bm e_r\|.
\]
Using $\varepsilon\le \|\bm e\|$ and $\|\bm e_t\|+\|\bm e_r\|\le \sqrt{2}\,\|\bm e\|$, we further obtain
\[
\frac{\lambda_t\varepsilon}{2\sqrt{2}}\,\|\bm e_t\|+\frac{\lambda_r\varepsilon}{2\sqrt{2}}\,\|\bm e_r\|
\le
\frac{\lambda}{2\sqrt{2}}\,\varepsilon\bigl(\|\bm e_t\|+\|\bm e_r\|\bigr)
\le
\frac{\lambda}{2}\,\|\bm e\|^2,
\] $\lambda:=\min\{\lambda_t,\lambda_r\}$. Therefore, whenever $\|\bm e\|\ge\varepsilon$,
\begin{align*}
\dot V(\bm e)\le&
-\lambda_t\|\bm e_t\|^2-\lambda_r\|\bm e_r\|^2+\frac{\lambda}{2}\|\bm e\|^2\\
\le&
-\lambda\|\bm e\|^2+\frac{\lambda}{2}\|\bm e\|^2
=
-\frac{\lambda}{2}\|\bm e\|^2.
\end{align*}
Thus $V(\bm e)$ decreases outside the $\varepsilon$-ball in $\bm e$. Standard comparison arguments, together with the quadratic bounds on $V$ from Assumption~\ref{ass:nominal_lyap}, imply practical exponential convergence to an $\varepsilon$-neighborhood of the origin.\hfill$\square$

\subsection{Implementation Procedure}\label{sec:procedure}
We conclude this section by summarizing the practical implementation steps of the proposed aggressiveness-aware learning-augmented controller.

\begin{enumerate}
\item \textit{State estimation and error computation.} At each control cycle, measure/estimate the current state $\x(t)$ (e.g., via IMU fusion and attitude/velocity estimation) and compute the tracking error $\bm e(t)$ w.r.t $\x_d(t)$.

\item \textit{Disturbance-label construction.} Construct a data-driven estimate $\y(t)\in\R^6$ of the generalized disturbance $\f(\x(t))$ by rearranging the dynamics in~\eqref{for:se3system}, i.e., subtracting known model contributions from measured/estimated accelerations and applied inputs. In practice, the required accelerations can be obtained from IMU-based estimates and/or filtered numerical differentiation of velocity and angular-rate estimates. Append the pair $(\x(t),\tilde{\y}(t))$ to the dataset, where $\tilde{\y}(t)=\y(t)+\bm\eta(t)$ and $\bm\eta(t)$ models measurement noise (zero-mean sub-Gaussian noise, with Gaussian noise $\mathcal N(0,\sigma^2 I_6)$ as a special case).

\item \textit{Oracle update (offline/online/hybrid).} On each dataset update interval $t\in[t_N,t_{N+1})$, freeze the dataset as $\mathcal D_N$ and train/update the oracle to obtain $\hat{\f}_N:\mathcal X\to\R^6$ (e.g., the GP posterior mean in Section~\ref{sec:gp}).

\item \textit{High-probability model-error bound.} Choose a compact set $\mathcal X_c$ capturing the expected closed-loop operating envelope (e.g., a tube around the reference trajectory) and compute a high-probability bound $\bar\rho_N(\cdot,\delta)$ on $\mathcal X_c$ such that, with probability at least $\delta$, the oracle error is uniformly bounded on $\mathcal X_c$ (cf. Assumption~\ref{ass:errorbound}). For GP oracles, $\bar\rho_N$ can be expressed in terms of posterior variances and information-gain quantities.

\item \textit{Aggressiveness-aware gain scheduling.} Select feedback gains $H_N\in\mathcal H$ to satisfy the tracking specification while reducing feedback-induced aggressiveness. In particular, enforce the sufficient condition~\eqref{eq:gain_condition} using conservative estimates/bounds for the Lyapunov constants $c_1,c_2$ in Assumptions~\ref{ass:nominal_lyap}--\ref{ass:gradV}. If separate translational and rotational error bounds are available, schedule the corresponding blocks of $H_N$ (Proposition~\ref{prop:block_sched}).

\item \textit{Learning-augmented control input.} Apply the learning-augmented control law~\eqref{for:framework:ctrl}. 
\end{enumerate}

\section{Numerical Results}\label{sec:sim}

We evaluate the proposed aggressiveness-aware learning-augmented control framework via numerical simulations. 
The goal is to illustrate the trade-off between tracking performance and feedback-induced aggressiveness, and to show how an aggressiveness-aware gain selection can achieve a prescribed tracking tolerance with reduced control variation.

\subsection{Simulation Setup}
The simulated quadrotor follows the model~\eqref{for:se3system} with standard small-scale UAV parameters (mass $m=1~\mathrm{kg}$ and diagonal inertia). 
The vehicle tracks a smooth reference trajectory with nontrivial translational excitation in $x$--$y$ and small altitude modulation, while maintaining constant yaw. 
We report the position tracking error norm $\|e_p(t)\|$ and use the tolerance $\varepsilon=0.1~\mathrm{m}$ as a practical tracking requirement over a simulation horizon of $T=20~\mathrm{s}$.

To emulate realistic model mismatch, we inject unknown external forces and moments that combine (i) drag-like terms, (ii) bounded lateral ``wind'' components, and (iii) persistent oscillatory disturbances in the vertical force and yaw moment channels, which require sustained control activity. 
The overall disturbance magnitude is scaled by a factor $\texttt{DIST\_SCALE}\in\{1,3\}$ to represent moderate and severe mismatch scenarios. Throughout the reported runs, the commanded inputs remained within the set $\mathcal U$.

We compare three gain-selection strategies:
(i) \emph{fixed-low} feedback gains (translation scaling $\texttt{trans\_scale}=1.0$),
(ii) \emph{fixed-high} feedback gains ($\texttt{trans\_scale}=1.8$), and
(iii) an \emph{aggressiveness-aware} selection that sweeps $\texttt{trans\_scale}\in[1.0,2.5]$ and picks the smallest value that satisfies $\|e_p(T)\|\le \varepsilon$ when feasible (otherwise returning the best-performing value in the tested grid).
In addition to the theoretical metric $s(H,\x)$, aggressiveness is quantified in simulation through observable input-variation metrics: the thrust-rate magnitude $|\dot T|$ and torque-rate norm $\|\dot{\tau}\|$. 
To highlight transient effects, we compute RMS values separately on a transient window $[0,3]~\mathrm{s}$ (shaded in gray in the plots) and on the steady window $[3,20]~\mathrm{s}$. 
\subsection{Performance--Aggressiveness Trade-off}
To emulate the gain-selection viewpoint of~\eqref{for:se3maxtrack} in a simulation setting, we implement a sweep-based scheduler: for a grid of values $\texttt{trans\_scale}\in[1.0,2.5]$, we simulate the closed loop and select the smallest gain that satisfies the practical tolerance $\|e_p(T)\|\le\varepsilon$ when feasible. In the theoretical development, feasibility is certified via a model-error bound $\bar\rho_N(\cdot,\delta)$; here we mimic this certification through the empirical sweep.

\paragraph{Moderate disturbance ($\texttt{DIST\_SCALE}=1$).}
Fig.~\ref{fig:perf} shows the performance--aggressiveness trade-off. Increasing feedback gains improves tracking: the final position error decreases from $\|e_p(T)\|=0.116~\mathrm{m}$ (fixed-low) to $0.060~\mathrm{m}$ (fixed-high), while the peak error remains nearly unchanged (within numerical resolution, $\max_t\|e_p(t)\|\approx 0.443~\mathrm{m}$). This improvement comes with increased transient aggressiveness, as shown in Fig.~\ref{fig:aggr}: the transient RMS rises from $|\dot T|_{\mathrm{RMS,tr}}=7.492~\mathrm{N/s}$ and $\|\dot{\tau}\|_{\mathrm{RMS,tr}}=5.777~\mathrm{Nm/s}$ (fixed-low) to $9.925~\mathrm{N/s}$ and $8.201~\mathrm{Nm/s}$ (fixed-high). In contrast, the steady-state torque-rate is nearly unchanged ($\|\dot{\tau}\|_{\mathrm{RMS,ss}}\approx 0.059~\mathrm{Nm/s}$ in both cases), indicating that the persistent oscillatory disturbances dominate the steady control modulation, whereas increased feedback mainly amplifies the transient response. The aggressiveness-aware scheduler selects $\texttt{trans\_scale}=1.2$ as the smallest gain satisfying $\varepsilon=0.1~\mathrm{m}$, yielding $\|e_p(T)\|=0.094~\mathrm{m}$ while reducing transient aggressiveness relative to fixed-high ($|\dot T|_{\mathrm{RMS,tr}}=8.121~\mathrm{N/s}$ and $\|\dot{\tau}\|_{\mathrm{RMS,tr}}=6.416~\mathrm{Nm/s}$).

\paragraph{Severe disturbance ($\texttt{DIST\_SCALE}=3$).}
When the disturbance magnitude is tripled, fixed-low gains fail to meet the tolerance, ending at $\|e_p(T)\|=0.351~\mathrm{m}$, while fixed-high gains improve tracking but still violate $\varepsilon$ with $\|e_p(T)\|=0.190~\mathrm{m}$. Within the tested grid, no gain value satisfies $\varepsilon=0.1~\mathrm{m}$; the sweep therefore returns the best-performing tested gain (largest $\texttt{trans\_scale}$), namely $\texttt{trans\_scale}=2.5$, resulting in $\|e_p(T)\|=0.134~\mathrm{m}$. Transient aggressiveness increases with gain magnitude (e.g., fixed-low versus fixed-high: $|\dot T|_{\mathrm{RMS,tr}}=7.887\rightarrow 10.411~\mathrm{N/s}$ and $\|\dot{\tau}\|_{\mathrm{RMS,tr}}=5.734\rightarrow 8.154~\mathrm{Nm/s}$), and the largest gain further increases torque-rate transients ($\|\dot{\tau}\|_{\mathrm{RMS,tr}}=10.886~\mathrm{Nm/s}$ at $\texttt{trans\_scale}=2.5$). Meanwhile, the steady torque-rate remains of comparable magnitude across strategies ($\|\dot{\tau}\|_{\mathrm{RMS,ss}}\approx 0.177~\mathrm{Nm/s}$). These results indicate that under severe mismatch, achieving strict tracking tolerances requires either improved disturbance compensation or accepting higher aggressiveness.

\begin{figure}[h!]
\centering
\begin{subfigure}{0.48\linewidth}
\centering
\includegraphics[width=\linewidth]{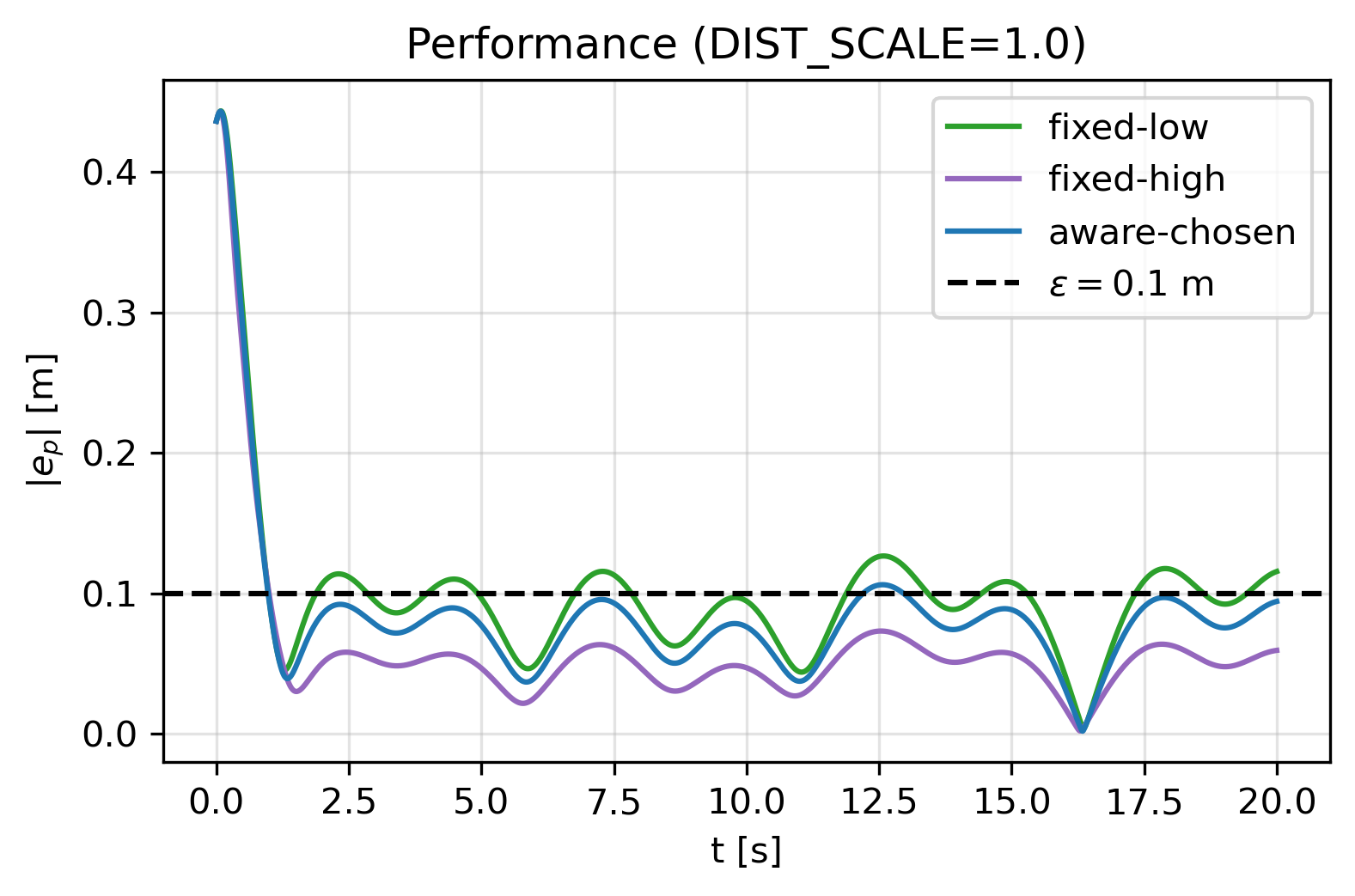}
\caption{$\texttt{DIST\_SCALE}=1$}
\end{subfigure}\hfill
\begin{subfigure}{0.48\linewidth}
\centering
\includegraphics[width=\linewidth]{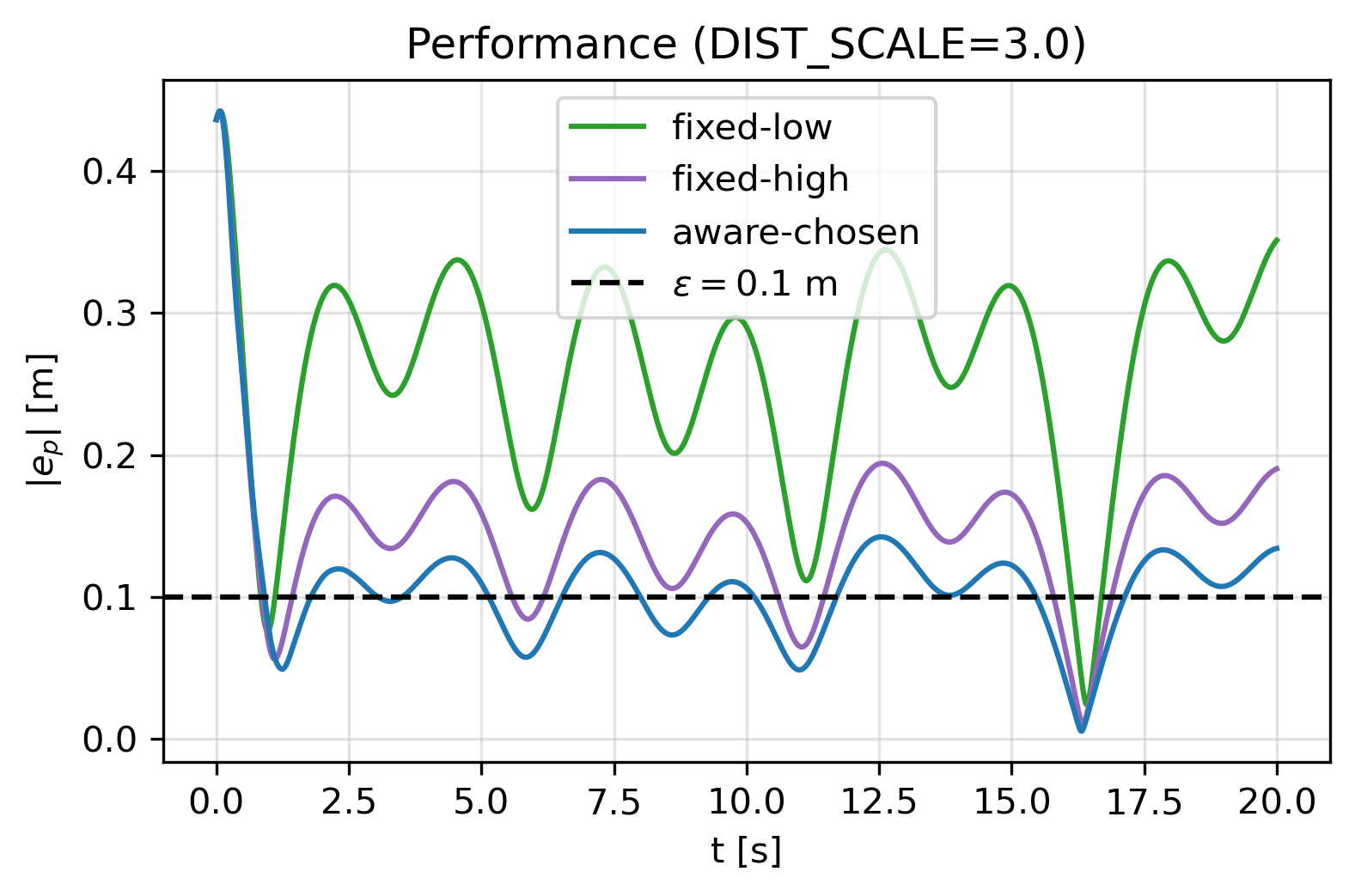}
\caption{$\texttt{DIST\_SCALE}=3$}
\end{subfigure}
\caption{Tracking performance under fixed-low, fixed-high, and aggressiveness-aware gain selection. The dashed line denotes the tolerance $\varepsilon=0.1~\mathrm{m}$.}
\label{fig:perf}
\end{figure}

\begin{figure}[h!]
\centering
\begin{subfigure}{0.48\linewidth}
\centering
\includegraphics[width=\linewidth]{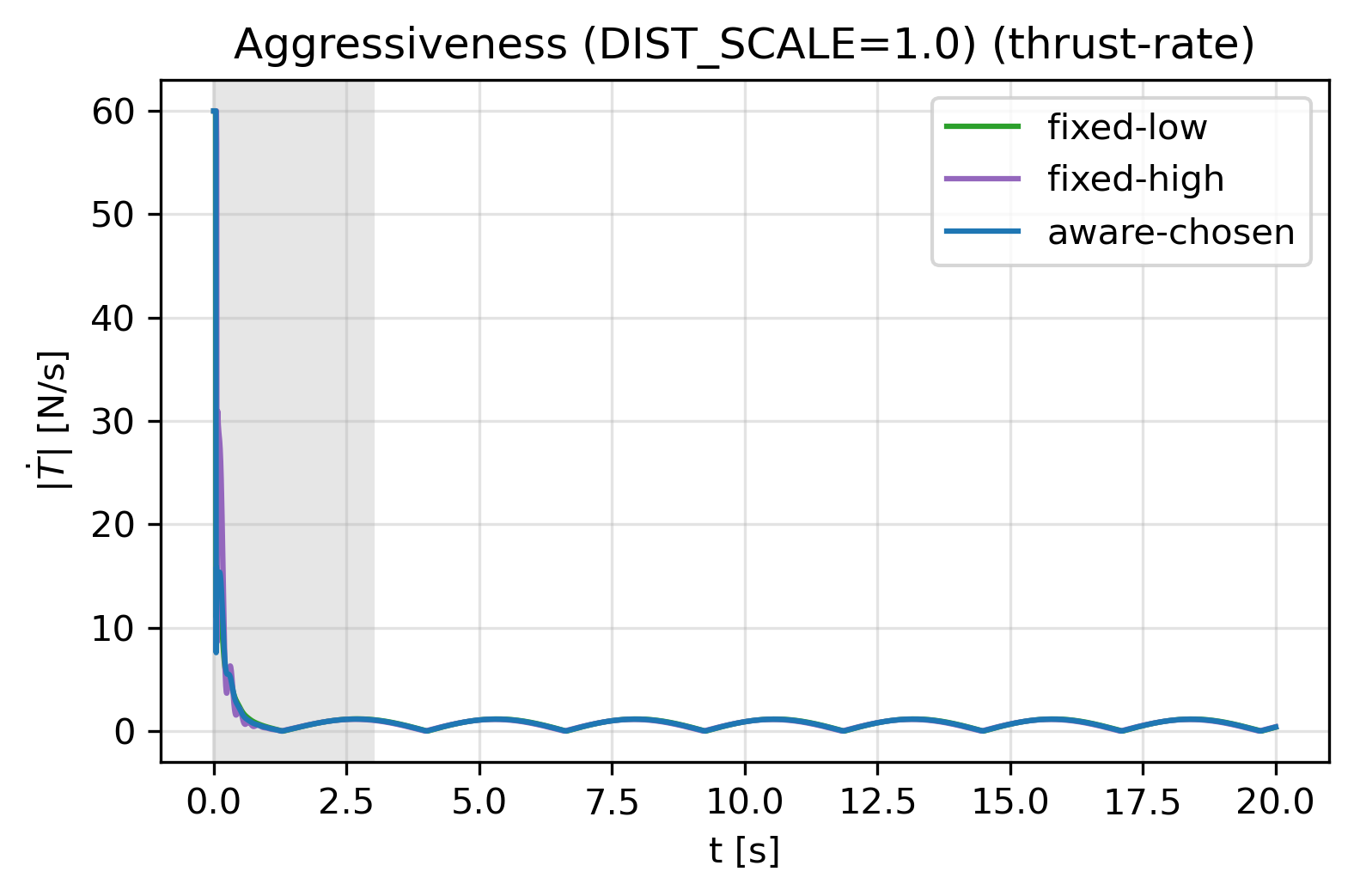}
\caption{$|\dot T|$, $\texttt{DIST\_SCALE}=1$}
\end{subfigure}\hfill
\begin{subfigure}{0.48\linewidth}
\centering
\includegraphics[width=\linewidth]{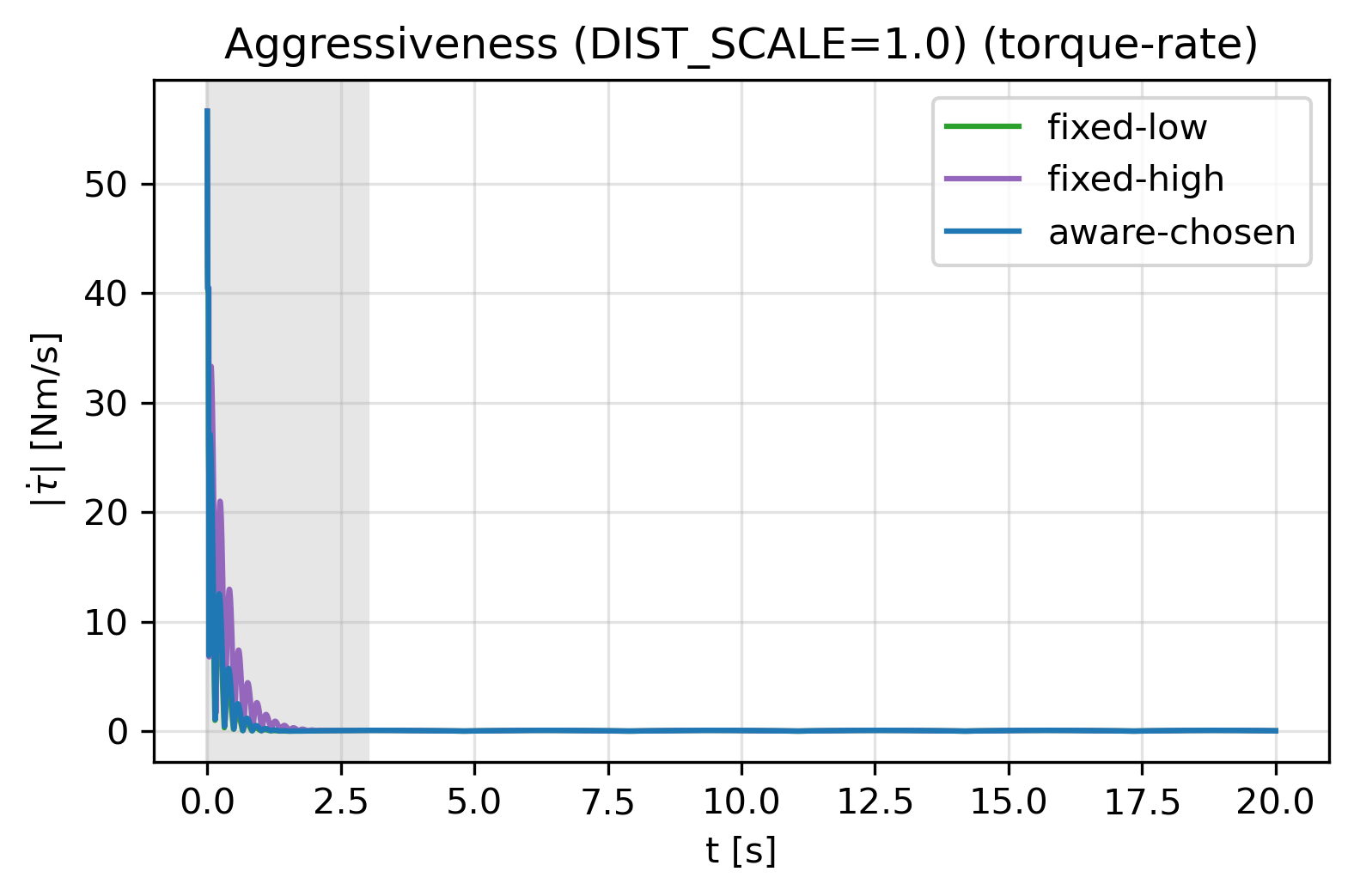}
\caption{$\|\dot\tau\|$, $\texttt{DIST\_SCALE}=1$}
\end{subfigure}

\vspace{2mm}

\begin{subfigure}{0.48\linewidth}
\centering
\includegraphics[width=\linewidth]{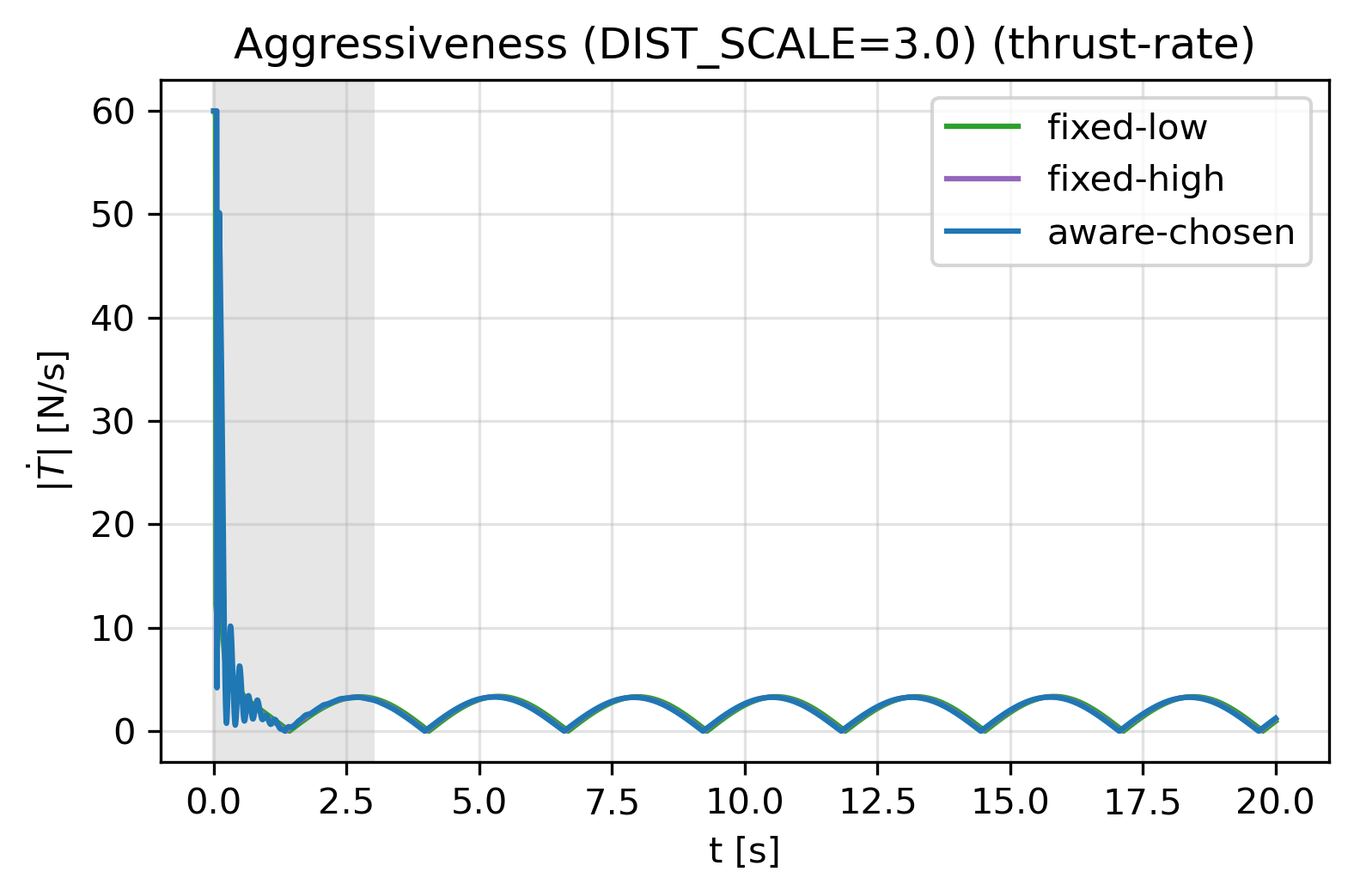}
\caption{$|\dot T|$, $\texttt{DIST\_SCALE}=3$}
\end{subfigure}\hfill
\begin{subfigure}{0.48\linewidth}
\centering
\includegraphics[width=\linewidth]{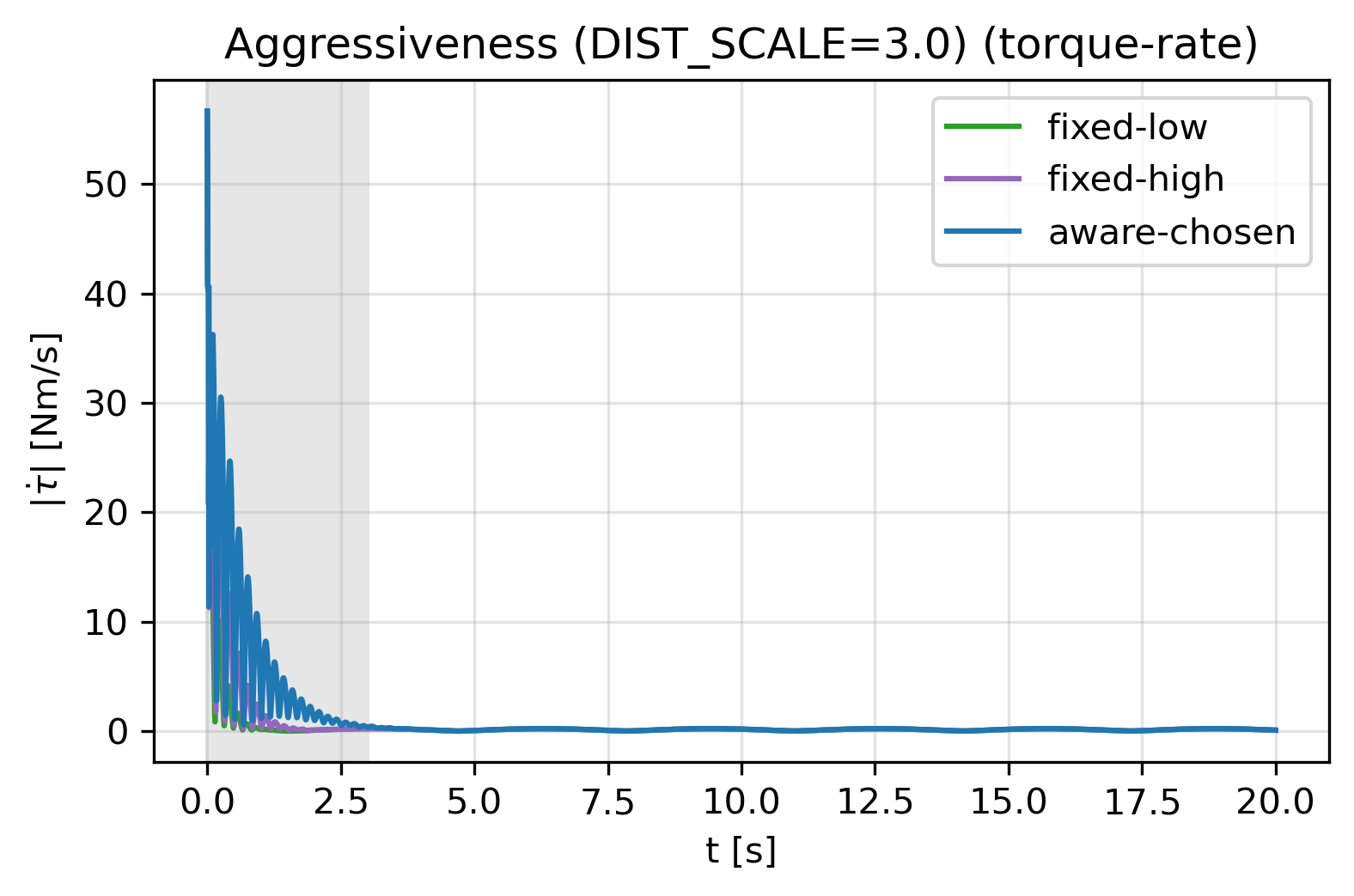}
\caption{$\|\dot\tau\|$, $\texttt{DIST\_SCALE}=3$}
\end{subfigure}
\caption{Aggressiveness metrics under fixed-low, fixed-high, and aggressiveness-aware gain selection. The shaded region marks the transient window $[0,3]~\mathrm{s}$ used for transient RMS.}
\label{fig:aggr}
\end{figure}



\subsection{GP-based disturbance compensation}
\label{subsec:offline}

We validate the theoretical results through numerical simulations under parametric mismatch and exogenous disturbances, and explicitly connect the observed feasibility--aggressiveness trade-off to the probabilistic bounds. The disturbance oracle learns a $6$-dimensional generalized disturbance (force + torque) using $6$ independent GP regressors. Each regressor maps a $20$-dimensional feature vector
$z = [p, v, q, \omega, \sin(\cdot), \cos(\cdot), \texttt{DIST\_SCALE}]$
to one component of $[f^\top,\ \tau_d^\top]^\top\in\mathbb{R}^6$.
Training data are collected from simulations at $\texttt{DIST\_SCALE}\in\{1,3\}$ with target normalization by disturbance scale enabled (\texttt{normalize\_by\_dist=True}), yielding $N_{\mathcal D}=1162$ samples with feature matrix $Z\in\mathbb{R}^{1162\times 20}$ and target matrix $Y\in\mathbb{R}^{1162\times 6}$.

The GP prior uses a squared-exponential (a.k.a.\ radial-basis-function, RBF) covariance with automatic relevance determination (ARD), i.e., one characteristic length-scale per input coordinate, and an additive i.i.d.\ white-noise term. Concretely, for each output channel $j\in\{1,\dots,6\}$ we fit an independent GP with kernel
\[
k_j(z,z') \;=\; \sigma_{f,j}^2\,
\exp\!\Big(-\tfrac12\sum_{i=1}^{20}\frac{(z_i-z'_i)^2}{\ell_{j,i}^2}\Big)
\;+\;\sigma_{n,j}^2\,\delta_{z,z'} ,
\]
where $\sigma_{f,j}^2$ is the signal variance (kernel amplitude), $\ell_{j,i}$ are the ARD length-scales, and $\sigma_{n,j}^2$ is the observation-noise variance. For the trained model, the learned kernel reported in the log (output~1) is $k_1(z,z') \;=\; (0.928)^2\,\mathrm{SE}\text{-}\mathrm{ARD}(z,z';\ell)
\;+\;(10^{-6})\,\delta_{z,z'}$, with the optimized ARD length-scales $\ell$
corresponding to the 20-dimensional feature vector used for GP regression. Large values (near the upper bound) indicate weak sensitivity of the posterior mean to the associated feature, whereas smaller values identify the most informative coordinates for predicting the disturbance residual.

To ensure that GP-based compensation is applied only when the model is confident, the raw GP mean is further modulated by an uncertainty-based gate. Specifically, at each time step we compute a scalar uncertainty $\rho(t)$ as the Euclidean norm of the predicted standard deviations across the six outputs, and apply a smooth sigmoid gate $g(\rho)\in[0,1]$ (with reported steady-state average gate) so that compensation is suppressed when epistemic uncertainty is high. The gated compensation is additionally saturated and low-pass filtered to prevent spurious high-frequency actuation. The GP-comp controller restores feasibility under severe mismatch: the learning-augmented loop achieves $\|e_p(T)\|=0.028\,\mathrm{m}$ (well below $\varepsilon=0.10\,\mathrm{m}$) while keeping aggressiveness metrics comparable to the non-learning baselines (e.g., $|\dot{T}|_{\mathrm{RMS,tr}}=8.066\,\mathrm{N/s}$, $\|\dot{\tau}\|_{\mathrm{RMS,tr}}=6.077\,\mathrm{Nm/s}$; steady-state $|\dot{T}|_{\mathrm{RMS,ss}}=2.251\,\mathrm{N/s}$, $\|\dot{\tau}\|_{\mathrm{RMS,ss}}=0.175\,\mathrm{Nm/s}$). The steady-state gate mean remains high ($\texttt{gate\_mean(ss)}=0.930$), indicating consistent use of the learned compensation where the controller must overcome persistent oscillatory disturbances. 

Figure~\ref{fig:offline_dist3} reports tracking for three controllers: (i) a conservative baseline (\texttt{fixed-low}), (ii) a more aggressive baseline (\texttt{fixed-high}), and (iii) the proposed GP-compensated controller with aggressiveness-aware gain selection (\texttt{GP-comp (aware)}). In this regime, \texttt{fixed-low} fails to meet the tracking requirement $\varepsilon=0.10$~m, ending at $\|e_p(T)\|=0.351$~m, and even \texttt{fixed-high} remains outside the tolerance with $\|e_p(T)\|=0.190$~m. In contrast, GP compensation restores feasibility with $\|e_p(T)\|=0.028$~m while selecting the \emph{minimum} translation gain scale that satisfies the tolerance in the offline sweep (here, $\texttt{trans\_scale}=1.0$). Quantitative aggressiveness and effort metrics are reported in Table~\ref{tab:metrics_dist3}.

\begin{figure}[h!]
  \centering
  \includegraphics[width=0.47\textwidth]{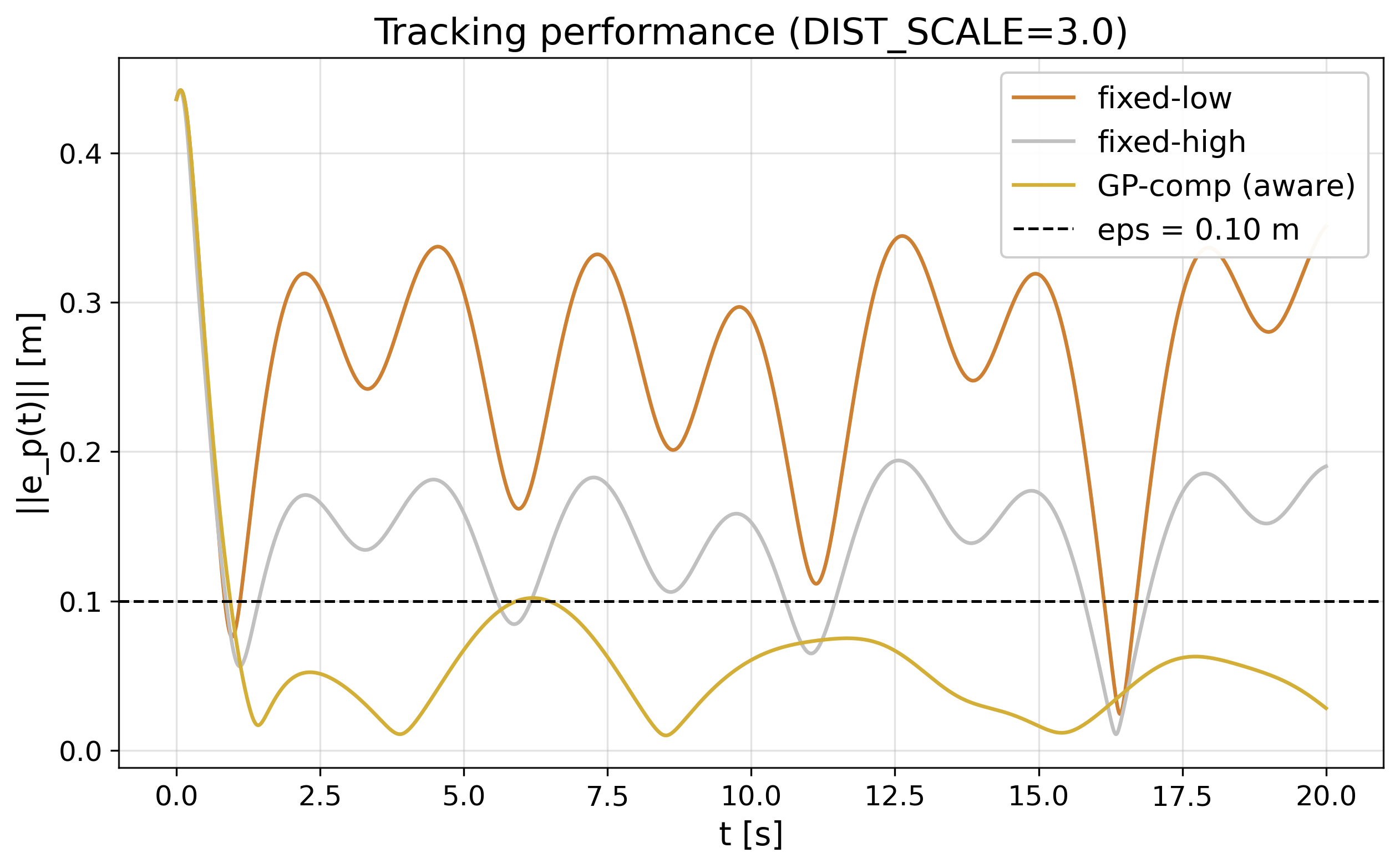}
  \caption{Numerical validation under severe mismatch ($\texttt{DIST\_SCALE}=3$): tracking performance for \texttt{fixed-low}, \texttt{fixed-high}, and \texttt{GP-comp (aware)}. The dashed line indicates the tracking tolerance $\varepsilon=0.10$~m.}
  \label{fig:offline_dist3}
\end{figure}

\begin{table}[h!]
\centering
\begin{tabular}{lccc}
\toprule
Metric & fixed-low & fixed-high & GP-comp (aware) \\
\midrule
$\|e_p(T)\|$ [m] & 0.351 & 0.190 & 0.028 \\
\addlinespace
$|\dot T|_{\mathrm{RMS,tr}}$ [N/s] & 7.887 & 10.411 & 8.066 \\
$\|\dot\tau\|_{\mathrm{RMS,tr}}$ [Nm/s] & 5.734 & 8.154 & 6.077 \\
$|\dot T|_{\mathrm{RMS,ss}}$ [N/s] & 2.327 & 2.290 & 2.251 \\
$\|\dot\tau\|_{\mathrm{RMS,ss}}$ [Nm/s] & 0.177 & 0.177 & 0.175 \\
\addlinespace
$|T-mg|_{\mathrm{RMS,tr}}$ [N] & 1.813 & 1.844 & 1.819 \\
$\|\tau\|_{\mathrm{RMS,tr}}$ [Nm] & 0.250 & 0.373 & 0.252 \\
$|T-mg|_{\mathrm{RMS,ss}}$ [N] & 2.002 & 1.980 & 1.958 \\
$\|\tau\|_{\mathrm{RMS,ss}}$ [Nm] & 0.172 & 0.172 & 0.170 \\
\addlinespace
$\|H\|_{\mathrm{F}}$ [--] & 17.866 & 26.961 & 17.866 \\
gate$_{\mathrm{mean,ss}}$ [--] & -- & -- & 0.930 \\
$\rho_{\mathrm{mean,ss}}$ [--] & -- & -- & 0.069 \\
\bottomrule
\end{tabular}
\caption{Numerical metrics for GP-based disturbance compensation with $\texttt{DIST\_SCALE}=3$. ``tr'' denotes $t\in[0,3]$~s and ``ss'' $t\in[3,T]$. Effort metrics use $T-mg$ and $\|\tau\|$.}
\label{tab:metrics_dist3}
\end{table}

\subsection{Online learning performance}
\label{subsec:online}

We additionally report a representative online experiment to illustrate the behavior of the scheme when the disturbance oracle is refined during execution. The GP model is first trained offline and then updated online using streaming data and a gating/activation mechanism to prevent unsafe corrections when uncertainty is high. We compare:
(i) \emph{fixed-low} gains without learning (baseline),
(ii) \emph{offline GP} compensation using a GP prior trained offline, and
(iii) \emph{offline+online} where a residual GP is updated online on top of the offline prior.
In all cases, we report the position tracking error $\|e_p(t)\|$ and the aggressiveness metrics $|\dot T|$ and $\|\dot\tau\|$, with transient and steady-state RMS computed on the same windows used throughout the paper, namely $[0,3]$\,s and $[3,20]$\,s.

The offline dataset was collected over two disturbance levels \texttt{DIST\_SCALE}$\in\{1,3\}$ and yields
$Z\in\mathbb{R}^{1162\times 20}$ input features and $Y\in\mathbb{R}^{1162\times 6}$ targets (3 force + 3 moment channels).
A six-output GP model was trained (independent GPs per output). To limit the effect of potentially unreliable online updates, the residual GP correction is gated based on an uncertainty measure. In the reported run, the online residual GP acquired \texttt{online\_points}$=350$ samples, and the gate statistics over steady-state were
\texttt{gate\_off\_mean(ss)}$=0.911$ and \texttt{gate\_on\_mean(ss)}$=0.982$, indicating that (i) the offline prior is trusted most of the time and (ii) the online residual is activated primarily when the uncertainty drops.

The baseline \emph{fixed-low} controller fails to satisfy the tolerance under severe mismatch, ending at $\|e_p(T)\|=0.351$\,m ($\max_t\|e_p(t)\|=0.442$\,m). 
In contrast, \emph{offline GP} compensation satisfies the tolerance with $\|e_p(T)\|=0.028$\,m, and adding the \emph{offline+online} residual update yields a small further improvement to $\|e_p(T)\|=0.025$\,m while preserving essentially the same aggressiveness levels. 
In particular, the transient RMS values remain comparable ($|\dot T|_{\mathrm{RMS,tr}}=8.066$\,N/s in both cases, and $\|\dot\tau\|_{\mathrm{RMS,tr}}=6.077\rightarrow 6.097$\,Nm/s), with steady-state rates essentially unchanged.

These outcomes support the core message of the theory: improving the model/oracle reduces the residual perturbation entering the error dynamics, thereby enabling practical exponential tracking with a prescribed tolerance while avoiding the need to increase feedback gains. In particular, for \texttt{DIST\_SCALE}$=3.0$ the learning-augmented controller restores feasibility of the $\varepsilon$-tracking requirement that is violated by low-gain feedback alone.

\begin{figure}[t]
  \centering
  \includegraphics[width=0.97\linewidth]{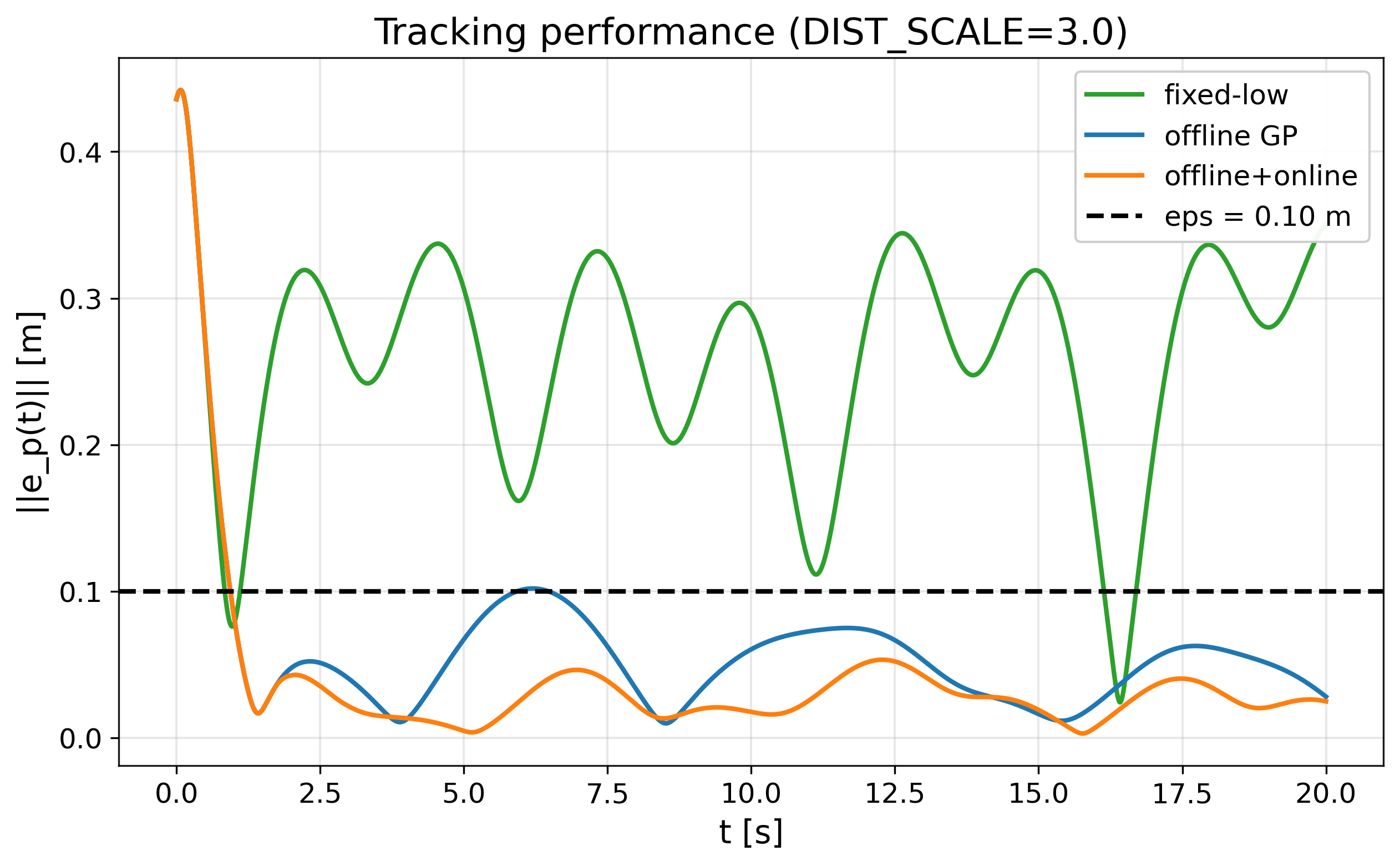}
  \caption{Online experiment, \texttt{DIST\_SCALE}$=3.0$. Position tracking error $\|e_p(t)\|$ for fixed-low (no GP), offline GP compensation, and offline+online residual GP. The dashed line denotes the tolerance $\varepsilon=0.10$\,m.}
  \label{fig:online_perf_dist3}
\end{figure}

Future work will focus on hardware experiments on a real quadrotor platform to validate the proposed aggressiveness-aware learning controller under realistic aerodynamic effects, sensing noise, and onboard computational constraints. We will investigate how to calibrate the uncertainty bounds and safety gates online from flight data so that prescribed tracking tolerances can be met with reduced control variation in practice under real-world conditions. The simulations presented here are intentionally designed to isolate the performance--aggressiveness mechanism; the experimental validation will include input saturation, discrete-time implementation, and state-estimation noise.



\bibliographystyle{IEEEtran}
\bibliography{root}

@inproceedings{lee2010geometric,
  author    = {Lee, Taeyoung and Leok, Melvin and McClamroch, N. Harris},
  title     = {Geometric Tracking Control of a Quadrotor UAV on {SE}(3)},
  booktitle = {Proceedings of the 49th IEEE Conference on Decision and Control (CDC)},
  year      = {2010},
  pages     = {5420--5425},
  address   = {Atlanta, GA, USA},
  publisher = {IEEE}
}

@article{beckers2022safe,
  title={Safe trajectory tracking for underactuated vehicles with partially unknown dynamics},
  author={Beckers, Thomas and Colombo, Leonardo J and Hirche, Sandra},
  journal={Journal of Geometric Mechanics},
  volume={14},
  number={4},
  pages={491--505},
  year={2022},
  publisher={Journal of Geometric Mechanics}
}

@inproceedings{Srinivas2010GPUCB,
  author    = {Niranjan Srinivas and Andreas Krause and Sham M. Kakade and Matthias Seeger},
  title     = {Gaussian Process Optimization in the Bandit Setting: No Regret and Experimental Design},
  booktitle = {Proceedings of the 27th International Conference on Machine Learning (ICML)},
  year      = {2010},
  pages     = {1015--1022}
}

@article{colombo2023learning,
    author={Colombo, Leonardo J. and Giribet, Juan I.},
  journal={IEEE Transactions on Control Systems Technology}, 
  title={Learning-Based Fault-Tolerant Control for an Hexarotor With Model Uncertainty}, 
  year={2024},
  volume={32},
  number={2},
  pages={672-679},
  doi={10.1109/TCST.2023.3318855}}

@INPROCEEDINGS{11312152,
  author={Beckers, Thomas and Colombo, Leonardo},
  booktitle={2025 IEEE 64th Conference on Decision and Control (CDC)}, 
  title={Physics-informed Learning for Passivity-based Tracking Control}, 
  year={2025},
  volume={},
  number={},
  pages={2091-2096},
  doi={10.1109/CDC57313.2025.11312152}}

@ARTICLE{11320442,
  author={Nieto, Omayra Yago and Simoes, Alexandre Anahory and Giribet, Juan I. and Colombo, Leonardo J.},
  journal={IEEE Transactions on Control of Network Systems}, 
  title={Learning-based Decentralized Control with Collision Avoidance for Multi-agent Systems}, 
  year={2025},
  volume={},
  number={},
  pages={1-12},
  doi={10.1109/TCNS.2025.3649723}}

@article{beckers2021online,
  title={Online learning-based trajectory tracking for underactuated vehicles with uncertain dynamics},
  author={Beckers, Thomas and Colombo, Leonardo J and Hirche, Sandra and Pappas, George J},
  journal={IEEE Control Systems Letters},
  volume={6},
  pages={2090--2095},
  year={2021},
  publisher={IEEE}
}

@article{nieto2024safe,
  title={Safe learning-based control for an aerial robot with manipulator arms},
  author={Yago Nieto, Omayra and Colombo, Leonardo J},
  journal={IFAC-PapersOnLine},
  volume={58},
  number={6},
  pages={36--41},
  year={2024},
  publisher={Elsevier}
}

@article{nieto2026dual,
  title={Dual-quaternion learning control for autonomous vehicle trajectory tracking with safety guarantees},
  author={Yago Nieto, Omayra and Anahory Simoes, Alexandre and Giribet, Juan I and Colombo, Leonardo},
  journal={arXiv preprint arXiv:2601.03097},
  year={2026}
}

@inproceedings{beckers2022learning,
  title={Learning-based Balancing of Model-based and Feedback Control for Second-order Mechanical Systems},
  author={Beckers, Thomas and Colombo, Leonardo J and Morari, Manfred and Pappas, George J},
  booktitle={2022 IEEE 61st Conference on Decision and Control (CDC)},
  pages={4667--4673},
  year={2022},
  organization={IEEE}
}

@article{lu2020robust,
  title={Robust autonomous flight in cluttered environment using a depth sensor},
  author={Lu, Liang and Yunda, Alexander and Carrio, Adrian and Campoy, Pascual},
  journal={International Journal of Micro Air Vehicles},
  volume={12},
  pages={1756829320924528},
  year={2020},
  publisher={SAGE Publications Sage UK: London, England}
}

@inproceedings{lopez2017aggressive,
  title={Aggressive 3-D collision avoidance for high-speed navigation.},
  author={Lopez, Brett Thomas and How, Jonathan P},
  booktitle={ICRA},
  pages={5759--5765},
  year={2017}
}

@article{BreedenGargPanagou2022CBFSampledData,
  author  = {Joseph Breeden and Kunal Garg and Dimitra Panagou},
  title   = {Control Barrier Functions in Sampled-Data Systems},
  journal = {IEEE Control Systems Letters},
  year    = {2022},
  volume  = {6},
  pages   = {367--372}
}

@ARTICLE{8291488,
  author={Franchi, Antonio and Carli, Ruggero and Bicego, Davide and Ryll, Markus},
  journal={IEEE Transactions on Robotics}, 
  title={Full-Pose Tracking Control for Aerial Robotic Systems With Laterally Bounded Input Force}, 
  year={2018},
  volume={34},
  number={2},
  pages={534-541},
  doi={10.1109/TRO.2017.2786734}}

@article{RashadBicegoZultSanchezEscalonillaJiaoFranchiStramigioli2022EnergyAwareImpedance,
  author  = {Ramy Rashad and Davide Bicego and Jelle Zult and Santiago Sanchez-Escalonilla and Ran Jiao and Antonio Franchi and Stefano Stramigioli},
  title   = {Energy Aware Impedance Control of a Flying End-Effector in the Port-Hamiltonian Framework},
  journal = {IEEE Transactions on Robotics},
  year    = {2022},
  volume  = {38},
  number  = {6}
}

@inproceedings{falanga2017aggressive,
  title={Aggressive quadrotor flight through narrow gaps with onboard sensing and computing using active vision},
  author={Falanga, Davide and Mueggler, Elias and Faessler, Matthias and Scaramuzza, Davide},
  booktitle={2017 IEEE international conference on robotics and automation (ICRA)},
  pages={5774--5781},
  year={2017},
  organization={IEEE}
}

@article{rodriguez2022autonomous,
  title={Autonomous aerial robot for high-speed search and intercept applications},
  author={Rodriguez-Ramos, Alejandro and Alvarez-Fernandez, Adrian and Bavle, Hriday and Rodriguez-Vazquez, Javier and Lu, Liang and Fernandez-Cortizas, Miguel and Fernandez, Ramon A Suarez and Rodelgo, Alberto and Santos, Carlos and Molina, Martin and others},
  journal={Field Robotics},
  volume={2},
  pages={1320--1350},
  year={2022},
  publisher={FRPS}
}

@article{DellaSantina2017SoftRobots,
  author  = {Cosimo Della Santina and Matteo Bianchi and Giorgio Grioli and Federico Angelini and Manuel G. Catalano and Manuel Garabini and Antonio Bicchi},
  title   = {Controlling Soft Robots: Balancing Feedback and Feedforward Elements},
  journal = {IEEE Robotics \& Automation Magazine},
  year    = {2017},
  volume  = {24},
  number  = {3},
  pages   = {75--83}
}

@article{Beckers2019StableGP,
  author  = {Thomas Beckers and Danijel Kuli{\'c} and Sandra Hirche},
  title   = {Stable Gaussian Process Based Tracking Control of Euler--Lagrange Systems},
  journal = {Automatica},
  year    = {2019},
  volume  = {103},
  pages   = {390--397}
}

@inproceedings{MellingerKumar2011MinSnap,
  author    = {Daniel Mellinger and Vijay Kumar},
  title     = {Minimum Snap Trajectory Generation and Control for Quadrotors},
  booktitle = {Proceedings of the IEEE International Conference on Robotics and Automation (ICRA)},
  year      = {2011},
  pages     = {2520--2525},
  publisher = {IEEE},
  doi       = {10.1109/ICRA.2011.5980409}
}

@article{MellingerKumar2012IJRR,
  author  = {Daniel Mellinger and Vijay Kumar},
  title   = {Trajectory Generation and Control for Precise Aggressive Maneuvers with Quadrotors},
  journal = {The International Journal of Robotics Research},
  year    = {2012},
  volume  = {31},
  number  = {5},
  pages   = {664--674},
  doi     = {10.1177/0278364911434236}
}

@inproceedings{Berkenkamp2015ECC,
  author    = {Felix Berkenkamp and Angela P. Schoellig},
  title     = {Safe and Robust Learning Control with Gaussian Processes},
  booktitle = {Proceedings of the European Control Conference (ECC)},
  year      = {2015}
}

@inproceedings{Berkenkamp2015SafeOpt,
  author    = {Felix Berkenkamp and Angela P. Schoellig and Andreas Krause},
  title     = {Safe Controller Optimization for Quadrotors with Gaussian Processes},
  booktitle = {Proceedings of the IEEE International Conference on Robotics and Automation},
  year      = {2016}
}

\end{document}